\documentclass[a4]{article}
\usepackage{a4wide}
\usepackage{epsfig}
\usepackage{graphicx}
\usepackage{amssymb}
\usepackage{amsfonts}
\usepackage{amsmath}
\usepackage{amsthm}
\usepackage{verbatim}
\usepackage{url}
\usepackage{stmaryrd}
\usepackage{algorithmic}
\usepackage{algorithm}
\usepackage[all]{xypic}
\usepackage{epic,eepic}
\usepackage{tikz}
\usetikzlibrary{arrows,shapes,decorations,automata,backgrounds,fit,matrix}
\usetikzlibrary{positioning}
\usetikzlibrary{snakes}
\usetikzlibrary{petri}
\usepackage{authblk}
\newtheorem{definition}{Definition}{\bfseries}{\itshape}
\newtheorem{lemma}{Lemma}{\bfseries}{\itshape}
\newtheorem{proposition}{Proposition}{\bfseries}{\itshape}
\newtheorem{corollary}{Corollary}{\bfseries}{\itshape}
\newtheorem{theorem}{Theorem}{\bfseries}{\itshape}
{\itshape}{\rmfamily}
\newenvironment{sketch}{\noindent \textit{Sketch of proof:}}{\qed}
\newtheorem{example}{Example}{\bfseries}{\itshape}

\newcommand{\nat}{\mathbb{N}}

\newcommand{\tuple}[1]{\langle #1 \rangle}
\newcommand{\set}[1]{\{ #1 \}}
\newcommand{\mset}[1]{\langle \langle#1 \rangle \rangle}

\newcommand{\vect}[1]{\mathbf{#1}}

\newcommand{\Interp}[1]{\llbracket #1 \rrbracket}
\newcommand{\Interpb}[1]{\llbracket #1 \rrbracket}           

\newcommand{\Rv}[1]{\mathit{RV}(#1)}
\newcommand{\Prot}{\mathcal{P}}
\newcommand{\Confs}{\mathcal{C}}
\newcommand{\Confseo}{\Gamma_{\even\odd}}

\newcommand{\Confsn}[1]{\mathcal{C}^{(#1)}}
\newcommand{\trans}{\rightarrow}
\newcommand{\transP}{\Rightarrow}
\makeatletter
\newcommand{\xRightarrow}[2][]{\ext@arrow 0359\Rightarrowfill@{#1}{#2}}
\makeatother
\newcommand{\ltransP}[1]{\xRightarrow{#1}}
\newcommand{\Net}{\mathcal{N}}
\newcommand{\Reach}[1]{\mathit{Reach}(#1)}

\newcommand{\even}{\textsf{E}}
\newcommand{\odd}{\textsf{O}}
 \newcommand{\eotrans}{\dashrightarrow}
 \newcommand{\leotrans}[1]{\overset{#1}{\dashrightarrow}}

\tikzstyle{pstate}=[rectangle, draw=black, rounded corners]

\tikzstyle{init-left}=[pin={[pin distance=7pt,pin
  edge={<-,black,thick},pin position=left]}]

\tikzstyle{state}=[circle,thick,draw=black,inner sep=0pt,minimum
    size=6mm]
\tikzstyle{finalstate}=[accepting,circle,thick,draw=black,inner sep=0pt,minimum
    size=6mm]

    \tikzstyle{lstate}=[diamond,thick,draw=black,inner sep=0pt,minimum
    size=6mm]
    \tikzstyle{lfinalstate}=[accepting,diamond,thick,draw=black,inner sep=0pt,minimum
    size=6mm]
    
\tikzstyle{statesbis}=[diamond,thick,inner sep=0pt,minimum size=6mm]
\tikzstyle{small-states}=[circle,thick,inner sep=0pt,minimum size=5mm]
\tikzstyle{small-statesbis}=[diamond,thick,inner sep=0pt,minimum size=5mm]
\tikzstyle{init-left}=[pin={[pin distance=7pt,pin edge={<-,black,thick}]left:}]
\tikzstyle{init-right}=[pin={[pin distance=7pt,pin edge={<-,black,thick}]right:}]
\tikzstyle{init-below}=[pin={[pin distance=7pt,pin edge={<-,black,thick}]below:}]
\tikzstyle{vertex}=[double,rounded corners,thick,inner sep=0pt,minimum
size=4mm]
\tikzstyle{vertexbis}=[diamond,double,thick,inner sep=0pt,minimum size=4mm]
\tikzstyle{petri-p}=[circle,thick,inner sep=0pt,minimum size=5mm]
\tikzstyle{petri-t}=[rectangle,thick,inner sep=0pt,minimum width=7mm,minimum height=1mm]
\tikzstyle{petri-t2}=[rectangle,thick,inner sep=0pt,minimum width=1mm,minimum height=7mm]
\tikzstyle{petri-tok}=[circle,inner sep=0pt,minimum size=4pt, color=black,fill=black]
\tikzstyle{gate}=[rectangle,thick,inner sep=0pt,minimum size=6mm,draw=black]
\title{Deciding the existence of cut-off in parameterized rendez-vous
  networks\thanks{Partly supported by ANR FREDDA (ANR-17-CE40-0013).}}

\author[1]{Florian Horn}
\author[1]{Arnaud Sangnier}
\affil[1]{IRIF, CNRS, Universite de Paris, France}
\date{}
\begin{document}
\maketitle            
\begin{abstract}
We study  networks of processes which all execute the same
finite-state protocol and communicate thanks to a rendez-vous
mechanism. Given a protocol, we are interested in checking whether
there exists a number, called a cut-off, such that in any networks with a bigger
number of participants, there is an execution where all the
entities end in some final states. We provide decidability and
complexity results of this problem under various assumptions, such as
absence/presence of a leader or symmetric/asymmetric rendez-vous.
\end{abstract}

\section{Introduction}

\textit{Networks with many identical processes.} One of the difficulty 
in verifying distributed systems lies in the fact that many of them are
designed for an unbounded number of participants. As a consequence, to
be exhaustive in the analysis,
one needs to design formal methods which takes into account this characteristic.
In \cite{german-reasoning-92}, German and Sistla introduce a model to
represent networks with a fix but unbounded number of entities. In this model, each participant 
executes the same protocol and they communicate between each other
thanks to rendez-vous (a synchronization mechanism allowing two
entities to change their local state simultaneously). The number of
participants can then be seen as a parameter of the model and possible
verification problems ask for instance whether a property holds for all the
values of this parameter or seeks for some specific value ensuring a
good behavior. With the
increasing presence of distributed mechanisms (mutual
exclusion protocols, leader election
algorithms, renaming algorithms, etc) in the core of our
computing systems, there
has been in the last two decades  a regain of attention in the study
of such parameterized networks.

Surprisingly, the verification of these parameterized systems is
sometimes easier than the case where the number of participants is
known. This can be explained by the following reason:  in the
parameterized case the procedure can adapt on demand the number of 
participants  to build a problematic execution. It is indeed what happens with the liveness verification of asynchronous
shared-memory systems. This problem is \textsc{Pspace}-complete for a finite
number of processes and in \textsc{NP} when this number is a parameter
\cite{durand-model-fmsd17}. It is hence worth studying the complexity
of the verification of such parameterized models and many recent works
have attacked these problems considering networks with different means of
communication. For instance in
\cite{eparza-verif-lics99,delzanno-param-concur10,bertrand-controlling-lmcs19,bertrand-reconfiguration-concur19}
the participants communicate thanks to broadcast of messages, in \cite{clarke-verification-concur04,aminof-param-vmcai14}
they use a token-passing mechanism , in \cite{bollig-param-rp14} a
message passing mechanism and in \cite{esparza-param-cav13} the
communication is performed through  shared registers. The relative expressiveness of some of those
models has been studied in~\cite{aminof-expressive-15}. Finally in his
survey~\cite{esparza-keeping-stacs14}, Esparza shows that minor
changes in the setting of parameterized networks, such as the presence of a
controller (or equivalently a leader), might drastically change the complexity 
of the verification problems. 

\smallskip

\noindent\textit{Cut-off to ease the verification.} When one has  to prove
the correctness of  a distributed algorithm designed to work for an
unbounded number of participants, one technique consists in proving
that the algorithm has a cut-off, i.e. a bound on the
number of processes such that if it behaves correctly for this
specific number of processes then
it will still be correct  for any bigger networks. Such a
property allows to reduce the verification procedure to the analysis
of the algorithm with a finite number of entities. Unfortunately, as
shown in \cite{aminof-param-dc18}, many parameterized systems do
not have a cut-off even for basic properties. Instead of checking whether a general class of models admits
a cut-off, we propose in this work to study the following problem: given a
representation of a system and a class of properties, does it admit a cutoff ?  To
the best of our knowledge, looking at the existence of a cutoff as a
decision problem is a subject that has not received a lot of attention
although it is interesting both practically and theoretically. First,
in the case where this problem is decidable,
it allows to find automatically cutoffs for specific systems
even though they belong to a class for which there is no general
results on  the existence of cutoff. The search of cutoffs has been studied in \cite{abdulla-parameterized-sttt16}
where the authors propose  a semi-algorithm for verification of parameterized
networks with respect to safety properties. This algorithm stops when
a cutoff is found. However it is not stated how to determine the
existence of this cutoff, neither if this is possible or not.  In
\cite{kaiser-dynamic-cav10}, the authors propose a way to compute
dynamically a cutoff, but they consider systems and properties for which they know
that a cutoff exists.
Second, from the theoretical point of view, the
cutoff decision problem is interesting because it goes beyond the
classical problems for parameterized systems that usually seek for the
existence of a number of participants which satisfies a property or
check that a property hold for all possible number of
participants. Note that in the latter case, one might be in a
situation that
for a property to hold a minimum number of participants is necessary
(and below this number the property does not hold), such a situation
can be detected with the existence of a cutoff but not with the simple
universal quantification.

\smallskip

\noindent \textit{Rendez-vous networks.} We focus on networks where the communication is performed by
rendez-vous. There are different reasons for this choice. First,  we
are not aware of any technique to decide automatically the existence
of a cut-off in parameterized  systems, it is hence convenient to look at
this problem  in a well-known setting. Another aspect which motivates the choice of this model is that
the rendez-vous communication corresponds
to a well-known paradigm in the design of concurrent/distributed
systems (for instance rendez-vous in the programming languages
\textsc{C} or \textsc{Java} can be easily implemented thanks to
wait/notify mechanisms). Rendez-vous communication seems as well a
natural feature  for parameterized systems used to model for instance
crowds or biological systems (at some point we consider
symmetric rendez-vous which can be seen less common in computing
systems but make sense for these other applications).  Last but not least, rendez-vous networks are very close to
population protocols \cite{angluin-computational-dc07} for which there has been in the last years a
regain of interest in the community of formal methods \cite{esparza-verification-AI17,blondin-peregrine-cav18,blondin-express-concur19}. Population
protocols and rendez-vous networks are both based on rendez-vous
communication,  but in population protocols it is furthermore required
that all the fair executions converge to some
accepting set of configurations (see \cite{esparza-verification-AI17}
for more details). In our case, we seek  for the
existence of an
execution ending with all the processes in a final state. The  similarities  between the two
models let us think that the formal techniques we use could be
adapted for
the analysis of some population protocols.

\smallskip

\noindent \textit{Our contributions.} We  study the
Cut-off Problem (C.O.P.) for rendez-vous networks. It consists in
determining whether, given a protocol labeled with rendez-vous
primitives, there exists a bound $B$, such that in any  networks of
size bigger than $B$ where the processes all run the same protocol there
is an execution which brings all the processes to a final
state. We assume furthermore that in our network,
there could be one extra entity, called the leader, that runs its own specific protocol. We first show that C.O.P. is decidable by
reducing it to a new decision problem on Petri nets. Unfortunately we
show as well that it is non elementary thanks to a reduction  from the
reachability problem in Petri nets\cite{czerwinski-reachability-stoc-19}. We then show that better
complexity bounds  can be obtained if we assume the rendez-vous to be
symmetric (i.e. any process that requests a rendez-vous can as well
from the same state  accept one and vice-versa) or if we
assume that there is no leader. For each of these restrictions, new
algorithmic techniques for the analysis of rendez-vous networks are
proposed. The following table sums up the complexity bounds we obtain.


\begin{table}[htbp]
 \begin{center}
  \begin{tabular}{|c||c|c|}
    \hline
    & Asymmetric rendez-vous &~~Symmetric rendez-vous~~\\
    \hline
    \hline
    ~~Presence of a leader~~&~~Decidable and non-elementary~~&\textsc{PSpace} \\
    \hline
    Absence of leader & \textsc{EXPSpace} & \textsc{NP}\\
    \hline
    \end{tabular}
    \caption{Complexity results obtained for the Cut-Off Problem}
  \end{center}
  \end{table}
Due to lack of space, omitted details and proofs can be found in Appendix.

\section{Modeling networks with rendez-vous communication}

We write $\nat$ to denote the set of natural numbers and $[i,j]$ to
represent the set $\set{k\in \nat \mid i\leq k \mbox{ and } k \leq
  j}$ for $i,j \in \nat$. For a finite set $E$, the set $\nat^E$ represents
the multisets over $E$. For two elements $m,m' \in  \nat^E$, we denote
$m+m'$ the multiset such that $(m+m')(e) = m(e) +m'(e)$ for
all $e \in E$. We say that $m \leq m'$ if and only if $m(e) \leq
m'(e)$ for all $e \in E$. If $m \leq m'$, then $m'-m$ is the multiset
such that  $(m'-m)(e) = m'(e)-m(e)$ for
all $e \in E$. The size of a multiset $m$ is given by
$|m| =\Sigma_{e\in E}m(e)$. For $e \in E$, we use sometimes the
notation $e$ for the multiset $m$ verifying $m(e)=1$ and
$m(e')=0$ for all $e' \in E\setminus\set{e}$ and 
the notation $\mset{ e1,e1,e2,e3}$ to represent the
multiset with four elements $e1,e1,e2$ and $e3$.

\subsection{Rendez-vous protocols}

We are now ready to define our model of networks. We assume that all
the entities in the network (called sometimes processes) behave
similarly following the same protocol except one entity, called the
leader, which might behave differently. The communication in the
network is pairwise and is performed by rendez-vous through a
communication alphabet $\Sigma$. Each entity can either request a
rendez-vous, with the primitive $?a$, or answer to a rendez-vous, with
the primitive $!a$ where $a$ belongs to $\Sigma$. The set of actions
is hence $\Rv{\Sigma}=\set{?a,!a\mid a \in \Sigma}$.

\begin{definition}[Rendez-vous protocol]
A \emph{rendez-vous protocol} $\Prot$ is a tuple
$\tuple{Q,Q_P,Q_L,\Sigma,q_i,q_f,\linebreak[0]q^L_i,q^L_f,E}$ where $Q$ is a finite set of
states partitioned into the processes states $Q_P$ and the leader
states $Q_L$, $\Sigma$ is a finite 
alphabet, $q_i\in Q_P$ [resp. $q^L_i\in Q_L$] is the initial state of the
processes [resp.  of the leader], $q_f\in Q_P$ [resp. $q^L_f\in Q_L$]
is the final state of the processes [resp. of the leader], and $E
\subseteq (Q_P \times \Rv{\Sigma} \times Q_P) \cup (Q_L \times
\Rv{\Sigma} \times Q_L)$ is the set of edges. 
\end{definition}

A configuration of the rendez-vous protocol $\Prot$ is a multiset $C \in
\nat^Q$ verifying that there exists $q \in Q_L$ such that $C(q)=1$ and $C(q')=0$
for all $q'\in Q_L\setminus\set{q}$, in other words there is a single
entity corresponding to the leader. The number of processes in a
configuration $C$ is given by $|C|-1$. We denote by $\Confsn{n}$
the set of configurations $C$ involving $n$ processes, i.e. such that $|C|=n+1$. The initial configuration
with $n$ processes $C^{(n)}_i$ is such that
$C^{(n)}_i(q_i)=n$ and $C^{(n)}_i(q^L_i)=1$ and  $C^{(n)}_i(q)=0$ for
all $q \in Q \setminus \set{q_i,q^L_i}$. Similarly the final  configuration
with $n$ processes $C^{(n)}_f$ verifies
$C^{(n)}_f(q_f)=n$ and $C^{(n)}_f(q^L_f)=1$ and  $C^{(n)}_f(q)=0$ for
all $q \in Q \setminus \set{q_f,q^L_f}$. Hence in an initial
configuration all the entities are in their initial state and in a
final configuration they are all in their final state. The notation $\Confs$ represents the whole set of configurations
equals to $\bigcup_{n \in \nat} \Confsn{n}$.

We are now
ready to formalize the behavior of a rendez-vous protocol. In this
matter, we define the relation $\trans \subseteq \bigcup_{n\geq1}
\Confsn{n}\times \Confsn{n}$ as follows : $C \trans C'$ if, and only
if, there is $a \in \Sigma$ and  two edges $(q_1,?a,q_2),
(q'_1,!a,q'_2) \in E$ such that $C(q_1)>0$ and $C(q'_1)>0$ and
$C(q_1)+C(q'_1)\geq 2$ and $C'=C -
(q_1+q'_1)+(q_2+q'_2)$. Intuitively it means that in $C$ there is one
entity in $q_1$ that requests a rendez-vous and one entity in $q'_1$
that answers to it and they both change their state to respectively
$q_2$ and $q'_2$. We need the hypothesis $C(q_1)+C(q'_1)\geq 2$ in
case $q_1=q'_1$. We use $\trans^\ast$ to represent the reflexive and
transitive closure of $\trans$. Note that if $C \trans^\ast C'$ then
$|C|=|C'|$, in other words there is no deletion or creation of
processes during an execution.

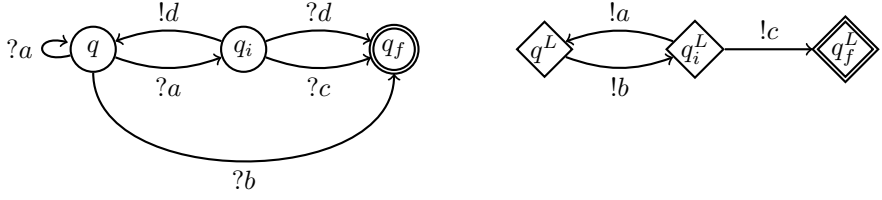
\begin{figure}[htbp]
\begin{center}
\scalebox{1}{
\begin{tikzpicture}[node distance=2cm]
\tikzstyle{every state}=[inner sep=3pt,minimum size=20pt]
\node(0)[state]{$q_i$};

\node(1)[finalstate,right of=0]{$q_f$}
edge [<-,thick,out=-160,in=-20] node[auto] {$?c$}(0)
edge [<-,thick,out=160,in=20] node[auto,swap] {$?d$}(0);

\node(2)[state,left of=0]{$q$}
edge [->,loop left,thick] node[auto] {$?a$}(2)
edge [->,thick,out=-90,in=-90] node[auto,swap] {$?b$}(1)
edge [<-,thick,out=20,in=160] node[auto] {$!d$}(0)
edge [->,thick,out=-20,in=-160] node[auto,swap] {$?a$}(0);

\node(4)[lstate,right of=1]{$q^L$};

\node(3)[lstate,right of=4]{$q^L_i$}
edge [->,thick,out=160,in=20] node[auto,swap] {$!a$}(4)
edge [<-,thick,out=-160,in=-20] node[auto] {$!b$}(4);

\node(5)[lfinalstate,right of=3]{$q^L_f$}
edge [<-,thick] node[auto,swap] {$!c$}(3);

\end{tikzpicture}
}
\end{center}
\caption{A rendez-vous protocol}
\label{fig:ex1}
\end{figure}

\begin{example}
  Figure \ref{fig:ex1} provides an example of rendez-vous
  protocol where the process states are represented  by
  circles and the leader states by diamond.
\end{example}

\subsection{The cut-off problem}

We can now describe the problem we
address. It consists in determining given a protocol
whether there exists a number of processes such that if we put more
processes in the network it is always possible to find an execution
which brings all the entities from their initial state to their final
state. This \textbf{cut-off problem (C.O.P.)} can be stated formally
as follows:
\begin{itemize}
\item \textbf{Input:} A rendez-vous protocol $\Prot$;
\item \textbf{Output:} Does there exist a cut-off $B \in \nat$ such
  that $C^{(n)}_i \trans^\ast C^{(n)}_f$ for
  all $n \geq B$ ?
\end{itemize}

\begin{example}
The rendez-vous network represented in Figure \ref{fig:ex1} admits a
cut-off equal to $3$.  For $n=3$, we have indeed an execution
$C^{(3)}_i \trans^\ast C^{(3)}_f$ : $\mset{q^L_i,q_i,q_i,q_i} \xrightarrow{d} \mset{q^L_i,q_i,q,q_f}
\xrightarrow{a} \mset{q^L, q_i,q,q_f} \xrightarrow{b} \mset{q^L_i, q_i,q_f,q_f}
\xrightarrow{c}  \mset{q^L_f, q_f,q_f,q_f}$ (we indicate for each
transition the label of the corresponding rendez-vous). For $n=4$, the
following sequence of rendez-vous leads to an execution $C^{(4)}_i
\trans^\ast C^{(4)}_f$: $\mset{q^L_i,q_i,q_i,q_i,q_i} \xrightarrow{d} \mset{q^L_i,q_i,q_i,q,q_f}
\xrightarrow{a} \mset{q^L, q_i,q_i,q_i,q_f} \xrightarrow{d} \mset{q^L, q_i,q,q_f,q_f}
\xrightarrow{b} \mset{q^L_i, q_i,q_f,q_f,\linebreak[0]q_f} \linebreak[0]\xrightarrow{c}  \mset{q^L_f,
q_f,q_f,q_f,q_f} $. Then for any $n >4$, we can always come back
to the case where $n=3$ (if $n$ is  odd) or $n=4$ (if $n$ is even). In
fact, we can always let $3$ or $4$ processes in $q_i$ and  move
pairwise the other processes, one in $q$ and one in
$q_f$. Then the processes in $q$ can be brought in $q_f$ thanks to the
rendez-vous $a$ and $b$ and the leader loop between $q^L_i$ and
$q^L$.  Note that if we delete the edge $(q,?a,q_i)$, this protocol
does not admit anymore a cut-off but for all odd number $n \geq 3$, we
have $C^{(n)}_i \trans^\ast C^{(n)}_f$. \end{example}

\subsection{Petri nets}

As we shall see there are some strong connections between rendez-vous
protocols and Petri nets, this is the reason why we recall the
definition of this latter model. 

\begin{definition}[Petri net]
A Petri net $\Net$ is a tuple $\tuple{P,T,Pre,Post}$ where $P$ is a
finite set of places, $T$ is a finite set of transitions, $Pre : T
\mapsto \nat^P$ is the precondition function and $Post : T\mapsto
\nat^P$ is the postcondition function.
\end{definition}

A marking of a Petri net is a multiset $M \in \nat^P$. A Petri net
defines a transition relation $\transP \subseteq \nat^P \times T
\times \nat^P$ such that $M \ltransP{t} M'$ for $M,M' \in \nat^P$ and
$t\in T$ if  and only if $M \geq Pre(t)$ and $M' = M -Pre(t) +
Post(t)$. The intuition behind Petri nets is that marking put tokens
in some places and each transition consumes with $Pre$ some tokens and
produces others thanks to $Post$ in order to create a new
marking. We write $M \transP M'$ iff there exists $t\in T$ such that $M \ltransP{t} M'$.  Given a marking $M \in \nat^P$, the reachability set of $M$
is the set $\Reach{M}=\set{M'\in \nat^P \mid M\transP^\ast M'}$ where
$\transP^\ast$ is the reflexive and transitive closure of $\transP$. One famous problem in Petri nets is the \textbf{reachability problem}:
\begin{itemize}
\item \textbf{Input:} A Petri net $\Net$ and two markings $M$ and $M'$;
\item \textbf{Output:} Do we have $M' \in \Reach{M}$ ?
\end{itemize}
This problem is decidable
\cite{mayr-algorithm-siam-84,kosaraju-decidability-stoc-82,lambert-structure-tcs-92,leroux-vector-popl-11}
and non elementary \cite{czerwinski-reachability-stoc-19}. Another similar
problem that we will refer to  and which  is easier to solve is the \textbf{reversible
  reachability problem}:
\begin{itemize}
\item \textbf{Input:} A Petri net $\Net$ and two markings $M$ and $M'$;
\item \textbf{Output:} Do we have $M' \in \Reach{M}$ and  $M \in \Reach{M'}$?
\end{itemize}
It has been shown in \cite{leroux-vector-lmcs-13} to be
\textsc{EXPSpace}-complete.

\section{Back and forth between rendez-vous protocols and Petri nets}
\label{sec:reduction}
\subsection{From Petri nets to rendez-vous protocols}
\label{sec:petritordv}

We will see here how the reachability problem for Petri nets can be
reduced to the C.O.P. which gives us a non-elementary
lower bound for this latter problem. We consider in the sequel a Petri
net $\Net=\tuple{P,T,Pre,Post}$ and two markings $M,M' \in
\nat^P$. Without loss of generality we can assume that  $M$ and $M'$ are of the
following form: there exists $p_i \in P$ such that $M(p_i)=1$ and $M(p)=0$ for
all $p \in P\setminus \set{p_i}$ and there exists
$p_f \in P$ such that $M'(p_f)=1$ and $M'(p)=0$ for all $p \in
P\setminus \set{p_f}$. Taking these restrictions on the
markings does not alter the complexity of the reachability problem.

\begin{figure}[htbp]
\begin{center}
\scalebox{1}{
\begin{tikzpicture}[node distance=1cm]
\tikzstyle{every state}=[inner sep=3pt,minimum size=20pt]
\node(0)[petri-p,draw=black,label=90:{\small $p_i$}]{};
\node(t1)[petri-t,draw=black,fill=black,below of=0,label=180:{\small
  $t_1$}]{}
edge [<-,thick](0);

\node(1)[petri-p,below left
of=t1,yshift=0.2cm,xshift=-0.5cm,draw=black,label=180:{\small
  $p_2$}]{}
edge [<-,thick](t1);
\node(2)[petri-p,below right of=t1,yshift=0.2cm,xshift=0.5cm,draw=black,label=0:{\small
  $p_3$}]{}
edge [<-,thick](t1);
\node(t2)[petri-t,draw=black,fill=black,below of=t1,label=180:{\small
  $t_2$}]{}
edge [<-,thick](1)
edge [<-,thick](2);;
\node(3)[petri-p,draw=black,below of=t2,label=-90:{\small $p_f$}]{}
edge [<-,thick](t2);

\node(r0)[lstate,right of=0,xshift=3cm]{$q^L_i$}
edge [->,thick,loop left] node[auto,scale=0.8] {$!a$}(r0)
;

\node(r1)[lstate,below of=r0,yshift=-1cm]{$q^L_s$}
edge [<-,thick] node[auto,pos=0.8,scale=0.8] {$!pr(p_i)$}(r0)
edge [->,thick,loop below] node[auto,scale=0.8] {$!b$}(r1);

\node(r2) [lstate,above left of=r1,xshift=-0.7cm]{}
edge [<-,thick] node[sloped,above,scale=0.8] {$!co(p_i)$}(r1);

\node(r3) [lstate,below left of=r1, xshift=-0.7cm]{}
edge [<-,thick] node[sloped,above,scale=0.8] {$!pr(p_2)$}(r2)
edge [->,thick] node[sloped,above,scale=0.8] {$!pr(p_3)$}(r1);;
\node(r4) [lstate,above right of=r1,xshift=0.5cm]{}
edge [<-,thick] node[sloped,above,scale=0.8] {$!co(p_2)$}(r1);
\node(r5) [lstate,below right of=r1, xshift=0.5cm]{}
edge [<-,thick] node[sloped,below,scale=0.8] {$!co(p_3)$}(r4)
edge [->,thick] node[sloped,above,scale=0.8] {$!pr(p_f)$}(r1);
\node(r6)[lfinalstate,below of=r1,yshift=-1cm]{$q^L_f$}
edge [<-,thick,in=-120,out=120] node[auto,scale=0.8,pos=0.4]
{$!co(p_f)$}(r1);

\node(s1)[state, right of=r0,xshift=3cm ]{$R$};
\node(s0)[state, below of=s1,yshift=-1cm]{$q_i$}
edge [->,thick] node[auto,pos=0.5,scale=0.8] {$?a$}(s1);
\node(s2) [state,left of=s1,xshift=-0.5cm]{$p_i$}
edge [<-,thick] node[sloped,above,scale=0.8] {$?pr(p_i)$}(s1);
\node(s3) [state,below left of=s1, yshift=-0.5cm]{$p_2$}
edge [<-,thick] node[sloped,above,scale=0.8] {$?pr(p_2)$}(s1);
\node(s4) [state,right of=s1,xshift=0.5cm]{$p_3$}
edge [<-,thick] node[sloped,above,scale=0.8] {$?pr(p_3)$}(s1);
\node(s5) [state,below right of=s1, yshift=-0.5cm]{$p_f$}
edge [<-,thick] node[sloped,above,scale=0.8] {$?pr(p_f)$}(s1);
\node(r6)[finalstate,below of=s0,yshift=-0.5cm]{$q_f$}
edge [<-,thick,in=-90,out=120] node[sloped,below,scale=0.8]
{$?co(p_2)$}(s3)
edge [<-,thick,in=-90,out=60] node[sloped,below,scale=0.8]
{$?co(p_f)$}(s5)
edge [<-,thick,in=-90,out=180] node[sloped,below,scale=0.8]
{$?co(p_i)$}(s2)
edge [<-,thick,in=-90,out=0] node[sloped,below,scale=0.8]
{$?co(p_3)$}(s4)
edge [<-,thick] node[auto,pos=0.8,scale=0.8]
{$?b$}(s0);

\node[petri-tok] at (0) {};
\end{tikzpicture}
}
\end{center}
\caption{A Petri net $\Net$ and its associated rendez-vous network $\Prot_\Net$}
\label{fig:petritordv}
\end{figure}
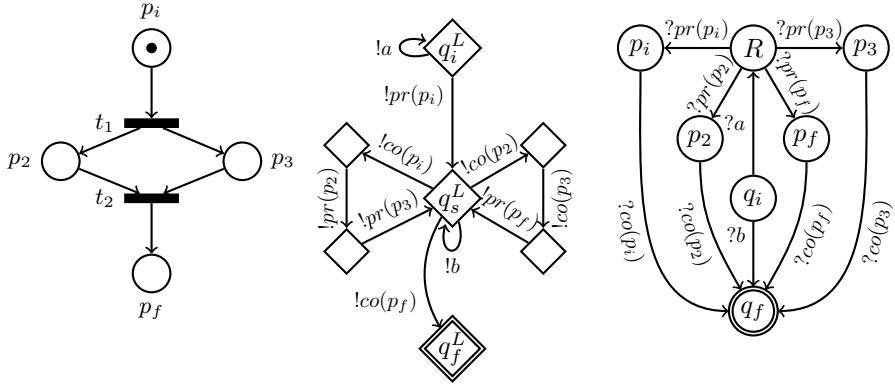

We build from $\Net$ a rendez-vous protocol $\Prot_\Net$ which admits
a cut-off if and only if $M' \in \Reach{M}$. The states of the  processes
in $\Prot_\Net$ are matched to the places of $\Net$, the number of
processes in a  state corresponding to  the number of tokens in
the associated place, and the leader is  in charge to move the
processes in order to simulate the changing on the number of tokens. The
protocol is equipped with an extra state $R$, the reserve state, where
the leader stores at the beginning of the simulation the number of
processes which will simulate the tokens: when a transition
produces a token in a place $p$, the leader moves a process from $R$
to $p$ and when it consumes a token from a place $p$, the  leader
moves a process from $p$ to $q_f$. Formally, we have: $\Prot_\Net=\tuple{Q,Q_P,Q_L,\Sigma,q_i,q_f,q^L_i,q^L_f,E}$ where:
\begin{itemize}
\item $Q_P = \set{q_i,q_f,R} \cup \set{p \mid p \in P}$,
\item  $Q_L = \set{q^L_i,q^L_f,q^L_s} \cup Q^{aux}_L$ (the states
  $Q^{aux}_L$ are extra states use by the leader while simulating
  transitions),
  \item $\Sigma=\set{a,b} \cup \set{co(p),pr(p) \mid p \in P}$,
\item $E \subseteq (Q_P \times \Rv{\Sigma} \times Q_P) \times (Q_L\times \Rv{\Sigma}  \times Q_L)$ is the smallest relation such that:
  \begin{itemize}
  \item $(q_i,?a,R) \in E$ and  $(q^L_i,!a,q^L_i) \in E$ (the leader send some processes in $R$),
  \item $(R,?pr(p),p) \in E$ and $(p,?co(p),q_f) \in E$  for all $p \in P$ (a production of a token moves a process from $R$ to $p$ and a consumption moves it from $p$ to $q_f$),
  \item $(q^L_i,!pr(p_i),q^L_s) \in E$ (the leader moves a process to $p_i$ and is in state $q^L_s$ where he simulates the transition),
  \item for each transition $t \in T$, there is in $E$ a sequence of edges: $(q^L_s,!co(p_1),q^t_1) \linebreak[0] (q^t_1,!co(p_2),q^t_2) ~~\ldots~~ \linebreak[0] (q^t_{k},!co(p_k),q^t_{k+1})(q^t_{k+1},!pr(p'_1),q^t_{k+2}) ~~\ldots ~~(q^t_{k+m-1}, \linebreak[0]!pr(p'_{m-1}) \linebreak[0],q^t_{k+m})(q^t_{k+m},!pr(p'_m),q^L_s)$ such that $Pre(t)=p_1+p_2+\ldots+p_k$ and $Post(t)=p'_1+p'_2+\ldots+p'_m$,
   \item  $(q_i,?b,q_f) \in E$ and $(q^L_s,?b,q^L_s) \in E$ (the leader can move the remaining processes in $q_f$),
   \item $(q^L_s,!co(p_f),q^L_f) \in E$ (the leader ends the simulation).
   \end{itemize}
\end{itemize}

Figure \ref{fig:petritordv}
provides an example of a Petri net and its associated
rendez-vous network. In this net, the transition letter $a$ is
used to put as many processes as necessary to simulate the number of
tokens in the places in the reserve state $R$. The letters $pr(p_j)$
are used to simulate the production of a token in the place $p_j$ by
moving a process from $R$ to $p_j$ and the letter $co(p_j)$ are used
to simulate the consumption of a token in the place $p_j$ by moving a
process from $p_j$ to $q_f$. It is then easy to see that each loop on the state $q^L_s$ simulates
a transition of the Petri net whereas the transition from $q^L_i$ to $q^L_s$
is used to build the initial marking and the transition from $q^L_s$
to $q^L_f$ is used to delete one token from the single place $p_f$
and move the corresponding process to $q_f$. Finally, the letter $b$ is used to ensure
the cutoff property by moving  from $q_i$ to $q_f$ the extra processes not needed to
simulate the tokens.
This construction ensures the following Lemma.

\begin{lemma}
  $M' \in \Reach{M}$ in $\Net$ iff there exists $B \in \nat$ such that for
  all $n \geq B$, we have $C^{(n)}_i \trans^\ast C^{(n)}_f$ in $\Net_\Prot$.
\end{lemma}

\begin{sketch}
  If $M' \in \Reach{M}$, then the cut-off is equal to $N+1$ where $N$ is the number of tokens produced during the execution from $M$ to $M'$. The leader first brings $N+1$ processes to $R$ thanks to the rendez-vous $a$ (the processes which remains in $q_i$ will be moved later from $q_i$ to $q_f$ thanks to the rendez-vous $b$). Then the leader moves to $q^L_s$ putting one process in $p_i$ (corresponding to one token in $p_i$) and from this state it simulates one by one the transitions of the execution by taking the corresponding loop on $q^L_s$. Each such loop simulates in fact a transition as follows:  it first consumes the tokens of the transition (by making processes move from a state $p$ to $q_f$) and then produces the corresponding tokens (by making processes move from $R$ a place $p$). When the leader has simulated all the transitions of the run, no more processes are in $R$, one process is in $p_f$ and some processes are left in $q_i$, the leader first empties $q_i$ (thanks to $b$) and then it moves the last process in $p_f$ to $q_f$ going himself to $q^L_f$.

  Assume now that there exists $n$ such that $C_i^{(n)} \trans^\ast C^{(n)}_f$ in $\Prot_\Net$. Then such an execution  is necessarily at each  step a move of the leader and of one process. According to the shape of the leader edges, we deduce that after having put some processes in $R$, it moves to $q^L_s$ where it will take a certain number of times some of the loops and finally it will move to $q^L_f$. Following the reverse reasonning as above this allows us to retrieve in the Petri net an execution from $M$ to $M'$ (each loop taken from $q^L_s$ corresponding to a fired transition).
\end{sketch}

\medskip

We can hence obtain a hardness result for the C.O.P. thanks to the fact
that the reachability problem in Petri nets is non-elementary \cite{czerwinski-reachability-stoc-19}.

\begin{theorem}
  \label{thm:nonelem}
The C.O.P. is non-elementary.
\end{theorem}

\subsection{From rendez-vous protocols to Petri nets}
\label{sec:topetri}

We now show how to encode the behavior of a rendez-vous protocol into
a Petri net and give a reduction from the C.O.P. to a problem on the
built Petri net. We consider a rendez-vous protocol
$\Prot=\tuple{Q,Q_P,Q_L,\Sigma,q_i,q_f,\linebreak[0]q^L_i,q^L_f,E}$. From
$\Prot$, we build a Petri net  $\Net_\Prot=\tuple{P,T,Pre,Post}$ with
the following characteristics~:

\begin{itemize}
\item $P=\set{p_q \mid q \in Q}$,
\item $T = \set{t_i,t^L_f} \cup \set{t_{(q_1,q_2,a,q_1',q'_2)} \mid q_1,q_2,q_1',q_2' \in Q \mbox{ and } a \in \Sigma \mbox{ and } (q_1,!a,q_1'),\linebreak[0](q_2,?a,q'_2) \in E }$,
\item the precondition function $Pre$ is such that:
  \begin{itemize}
  \item $Pre(t_i)(p)=0$ for all $p \in P$,
  \item $Pre(t^L_f)(p_{q^L_f})=1$ and $Pre(t^L_f)(p)=0$ for all $p \in P \setminus \set{p_{q^L_f}}$,
  \item $Pre( t_{(q_1,q_2,a,q_1',q'_2)})(p_{q_1})=Pre( t_{(q_1,q_2,a,q_1',q'_2)})(p_{q_2})=1$ and \\$Pre( t_{(q_1,q_2,a,q_1',q'_2)})(p) =0$ for all $p \in P \setminus \set{p_{q_1},p_{q_2}}$,
  \end{itemize}
 \item the postcondition function $Post$ is such that:
  \begin{itemize}
  \item $Post(t_i)(p_{q_i})=1$ and $Post(t_i)(p)=0$ for all $p \in P \setminus \set{p_{q_i}}$,
     \item $Post(t^L_f)(p)=0$ for all $p \in P$,
  \item $Post( t_{(q_1,q_2,a,q_1',q'_2)})(p_{q'_1})=Post( t_{(q_1,q_2,a,q_1',q'_2)})(p_{q'_2})=1$ and \\$Post( t_{(q_1,q_2,a,q_1',q'_2)})(p) =0$ for all $p \in P \setminus \set{p_{q'_1},p_{q'_2}}$.
  \end{itemize}
\end{itemize}

Intuitively in $\Net_\Prot$, we have a place for each state of $\Prot
$, the transition $t_i$ puts tokens corresponding to new
processes in the place corresponding to the initial state $q_i$, the
transition $t^L_f$ consumes a token in the place corresponding to the
final state of the leader $q^L_f$ and each transition
$t_{(q_1,q_2,a,q_1',q'_2)}$ simulates the
protocol respecting the associated semantics (it checks that there is
one process in $q_1$ another one in $q_2$ and that they can communicate thanks to the communication letter $a
\in \Sigma $ moving to $q'_1$ and $q'_2$). Figure \ref{fig:rdvtopetri} represents  the Petri net
$\Net_\Prot$ for the protocol $\Prot$ of Figure \ref{fig:ex1} (the
transitions are only labeled with the letter of the rendez-vous).

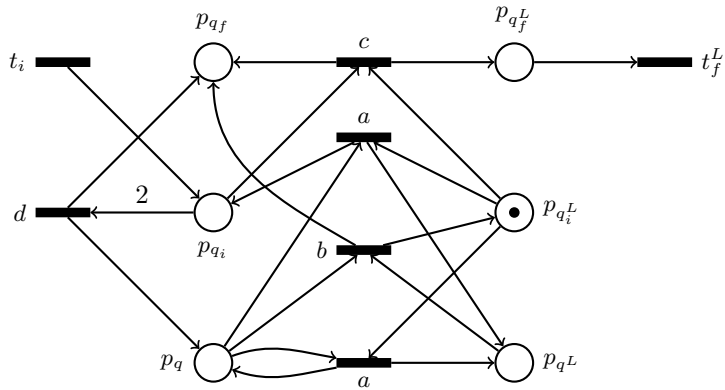
\begin{figure}[htbp]
\begin{center}
\scalebox{1}{
\begin{tikzpicture}[node distance=2cm]
\tikzstyle{every state}=[inner sep=3pt,minimum size=20pt]
\node(0)[petri-p,draw=black,label=90:{\small $p_{q_f}$}]{};
\node(1)[petri-p,below of=0,draw=black,label=-90:{\small $p_{q_i}$}]{};
\node(2)[petri-p,below of=1,draw=black,label=180:{\small $p_{q}$}]{};

\node(3)[petri-p,draw=black,right of=0,xshift=2cm,label=90:{\small $p_{q^L_f}$}]{};
\node(4)[petri-p,below of=3,draw=black,label=0:{\small $p_{q^L_i}$}]{};
\node(5)[petri-p,below of=4,draw=black,label=0:{\small $p_{q^L}$}]{};

\node(ti)[petri-t,draw=black,fill=black,left of=0,label=180:{\small $t_i$}]{}
edge [->,thick](1);
\node(t0)[petri-t,draw=black,fill=black,left of=1,label=180:{\small $d$}]{}
edge [<-,thick]node[auto] {$2$}(1)
edge [->,thick](0)
edge [->,thick](2);

\node(t1)[petri-t,draw=black,fill=black,right of=0,label=90:{\small $c$}]{}
edge [->,thick](0)
edge [->,thick](3)
edge [<-,thick](1)
edge [<-,thick](4);

\node(t2)[petri-t,draw=black,fill=black,right of=1,yshift=1cm,label=90:{\small $a$}]{}
edge [->,thick](1)
edge [<-,thick](4)
edge [<-,thick](2)
edge [->,thick](5);

\node(t3)[petri-t,draw=black,fill=black,right of=2,label=-90:{\small $a$}]{}
edge [<-,thick,in=20,out=170](2)
edge [->,thick,in=-20,out=-170](2)
edge [<-,thick](4)
edge [->,thick](5);

\node(t4)[petri-t,draw=black,fill=black,below of=t2,yshift=0.5cm,label=180:{\small $b$}]{}
edge [<-,thick](2)
edge [<-,thick](5)
edge [->,thick,in=-90,out=150](0)
edge [->,thick](4)
;
\node(tf)[petri-t,draw=black,fill=black,right of=3,label=0:{\small $t^L_f$}]{}
edge [<-,thick](3);

\node[petri-tok] at (4) {};
\end{tikzpicture}
}
\end{center}
\caption{The Petri net $\Net_\Prot$ for the protocol $\Prot$ of Figure \ref{fig:ex1}}
\label{fig:rdvtopetri}
\end{figure}

Unfortunately we did not find a way to reduce directly the C.O.P. to
the reachability problem in Petri nets which would have lead directly
to the decidability of C.O.P. However we will see how the C.O.P. on
$\Prot$ can lead to a decision  problem  on $\Net_\Prot $. We consider
the initial marking $M_0 \in \nat^P$ such that $M_0(p_{q^L_i})=1$ and
$M_0(p)=0$ for all $p \in P \setminus \set{p_{q^L_i}}$ and the family
of markings $(M^{(n)}_f)_{\set{n \in \nat}}$ such that
$M^{(n)}_f(p_{q_f})=n$ and $M^{(n)}_f(p)=0$ for all $p \in P \setminus
\set{p_{q_f}}$. From the way we build the Petri net $\Net_\Prot$, we
deduce  the following lemma:

\begin{lemma}
  \label{lem:reductionCOP}
  For all $n \in \nat$, $C^{(n)}_i \trans^\ast C^{(n)}_f$ in
  $\Prot$ iff $M^{(n)}_f \in \Reach{M_0}$ in $\Net_\Prot$.
\end{lemma}

This leads us to propose a cut-off problem for Petri nets, which asks
whether given an initial marking and a specific place, there exists a
bound $B \in \nat$ such that for all $n \geq B$ it is possible to
reach a marking with $n$ tokens in the specific place and none in the
other. This \textbf{single place cut-off problem (single place C.O.P.)} can be stated formally as follows:
\begin{itemize}
\item \textbf{Input:} A Petri net $\Net$, an initial marking $M_0$ and a place $p_f$;
\item \textbf{Output:} Does there exist $B \in \nat$ such that for
  all $n \geq B$, we have $M^{(n)} \in \Reach{M_0}$ in $\Net$ where $M^{(n)}$ is the marking verifying $M^{(n)}(p_f)=n$ and $M^{(n)}(p)=0$ for all $p \in P\setminus\set{p_f}$?
\end{itemize}

Thanks to Lemma \ref{lem:reductionCOP}, we can then conclude the following proposition which justifies the introduction of the single place C.O.P. in our context.

\begin{proposition}
  \label{prop:COPtosingleplace}
  The C.O.P. reduces to the single place C.O.P.
  
\end{proposition}

\section{Solving C.O.P. in the general case}
\label{sec:general}

We show how to solve the C.O.P. by solving the single place C.O.P.
To the best of our knowledge this latter problem has not yet been studied and we do not see direct connections with existing studied problems on Petri nets. It amounts to check if for some $B \in \nat$ we have  $\set{M \in \nat^P \mid M(p)=0 \mbox{ for all } p \in P\setminus\set{p_f} \mbox{ and } M(p_f)\geq B} \subseteq \Reach{M_0}$. We know from \cite{kleine-projections-tcs89} that the projection of the reachability set on the single place $p_f$ is semilinear (that can be represented by a Presburger arithmetic formula), however this does not help us since we furthermore require the other places different from $p_f$ to be empty.

\subsection{Formal tools and associated results}

For $\mathbf{P},\mathbf{P}' \subseteq \nat^n$, we let $\mathbf{P} +\mathbf{P}'=\set{p + p' \mid p\in \mathbf{P} \mbox{ and } p' \in \mathbf{P}'}$ and we shall sometimes identify an element $p \in \nat^n$ with the singleton $\set{p}$.
A  subset $\mathbf{P}$ of $\nat^n$ for $n>0$ is said to be \emph{periodic} iff $\mathbf{0} \in \mathbf{P}$ and $\mathbf{P} + \mathbf{P} \subseteq \mathbf{P}$. Such a periodic set $\mathbf{P}$ is \emph{finitely generated} if there exists a finite set of elements $\set{\vect{p}_1,\ldots,\vect{p}_k} \subset \nat^n$ such that $\mathbf{P}=\set{\lambda_1.\vect{p}_1+\ldots+\lambda_k.\vect{p_k} \mid \lambda_i \in \nat \mbox{ for all } i \in [1,k]}$. A \emph{semilinear set} of $\nat^k$ is then a finite union of sets of the form $\vect{b} + \mathbf{P}$ where $\vect{b} \in \nat^k$ and $\mathbf{P}$ is finitely generated. Semilinear sets are particularly useful tools because they are closed under the classical operations (union, complement and projection) and they  provide a finite representation  of infinite sets of vectors of naturals. Furthermore they can be represented by logical formulae expressed in Presburger arithmetic which  is  the decidable first-order theory of natural numbers with addition. A formula $\phi(x_1,\ldots,x_k)$ of Presburger arithmetic with free variables $x_1,\ldots,x_k$ defines a set $\Interp{\phi} \subseteq \nat^k$ given by $\set{\vect{v} \in \nat^k \mid \vect{v}\models \phi}$ (here $\models$ is the classical satisfiability relation for Presburger arithmetic and it holds true if the formula holds when replacing each $x_i$ by $\vect{v}[i]$). In \cite{ginsburg-semigroups-66}, it was proven that a set $S \subseteq \nat^k$ is semilinear iff there exists a Presburger formula $\phi$ such that $S = \Interp{\phi}$. Note that the set $\set{M \in \nat^P \mid M(p)=0 \mbox{ for all } p \in P\setminus\set{p_f}}$ has  a single interesting component, the other being $0$. In \cite{kleine-projections-tcs89}, to prove that the projection of the reachability set of a Petri net on a single place is semilinear, the authors need the following lemma.

\begin{lemma}\cite{kleine-projections-tcs89}\label{lem:onedim:semilinear}
Let $S \subseteq \nat$. If there exist $m,t \in \nat$ such that for all $s \in S$, $s \geq m$ implies $s+t \in S$, then $S$  is semilinear.
\end{lemma}

\noindent This allows us to deduce the following result on periodic subsets of $\nat$.

\begin{lemma}
 \label{lem:onedim:periodic} 
Every periodic subset $\mathbf{P} \subseteq \nat$ is semilinear.
\end{lemma}

\begin{proof}
If $\mathbf{P} = \emptyset$ or $\mathbf{P}=\set{0}$ then it is semilinear. Otherwise, let $m$ be the minimal strictly positive element of $\mathbf{P}$. Then for any $s \in \mathbf{P}$ such that $s \geq m$, since $\mathbf{P}$ is periodic,  we have $s+m \in \mathbf{P}$. By Lemma \ref{lem:onedim:semilinear}, we get that $\mathbf{P}$ is semilinear.
\end{proof}

We now recall some connections between Petri nets and semilinear sets. Let $\Net=\tuple{P,T,Pre,\linebreak[0]Post}$ be a Petri net with $P=\set{p_1,\ldots,p_k}$, this allows us to look at the markings as elements of $\nat^k$ or of $\nat^P$. Given a language of finite words of transitions $L \subseteq T^\ast$ and a marking $M$, let $\Reach{M,L}$ be the reachable markings produced by $L$ from $M$ defined by $\set{M' \subseteq \nat^k \mid \exists w \in L \mbox{ such that } M \ltransP{w} M'}$ where we extend in the classical way the relation $\transP$ over words of transitions by saying $M \ltransP{\varepsilon} M$ and if $w=t.w'$, we have $M \ltransP{w} M'$ iff there exists $M''$ such that $M \ltransP{t} M''\ltransP{w'} M'$. A flat expression of transitions is a regular expression over $T$ of the form $T_1T_2\ldots T_\ell$ where each $T_i$ is either a finite word in $T^\ast$ or of the form $w^\ast$ with $w \in T^\ast$. For a flat expression $FE$, we denote by $L(FE)$ its associated language. In \cite{fribourg-petri-wflp00}, the following result relating flat expressions of transitions and their produced reachability set is given (it has then been extended to more complex systems \cite{finkel-compose-fsttcs02}).

\begin{proposition}\cite{fribourg-petri-wflp00}\label{prop:flatpetri}
 Let $\Net=\tuple{P,T,Pre,Post}$ be a Petri net, $FE$ a flat expression of transitions and $M \in \nat^P$ a marking. Then $\Reach{M,L(FE)}$ is semilinear  (and the corresponding Presburger formula can be computed).
\end{proposition}

\subsection{Deciding if a bound is a single-place cut-off}

We prove that if one provides a bound $B \in \nat $, we are able to decide whether it corresponds to a cut-off as defined in the single place C.O.P. Let $\Net=\tuple{P,T,Pre,Post}$ be  a Petri net with an initial marking $M_0 \in \nat^P$, a specific place $p_f \in P$ and a bound $B \in \nat $. We would like to decide whether the following inclusion holds $\set{M \in \nat^P \mid M(p)=0 \mbox{ for all } p \in P\setminus\set{p_f} \mbox{ and } M(p_f)\geq B} \subseteq \Reach{M_0}$. An important point to decide this inclusion lies in the fact that the set $\set{M \in \nat^P \mid M(p)=0 \mbox{ for all } p \in P\setminus\set{p_f} \mbox{ and } M(p_f)\geq B}$ is semilinear and this allows us to use a method similar to the one proposed in \cite{jancar-cofinite-pn18} to check whether the reachability set of a Petri net equipped with a semilinear set of initial markings is universal. One key point is the following result which is a reformulation of a Lemma in  \cite{leroux-presburger-lics13}. This result was originally stated for Vector Addition System with States (VASS), but it is well known that a Petri net can be translated into a VASS with an equivalent reachability set.

\begin{proposition}\cite[Theorem 1]{jancar-cofinite-pn18}
  \label{prop:semilinear:flat}
 Let $\Net=\tuple{P,T,Pre,Post}$ be a Petri net, $M \in \nat^P$ a marking and $S \subseteq \nat^P$ a semilinear set of markings. If $S \subseteq \Reach{M}$ then there is a flat expression $FE$ of transitions such that $S \subseteq \Reach{M,L(FE)}$.
\end{proposition}

Following the technique used in \cite{jancar-cofinite-pn18}, this proposition provides us a tool to solve our inclusion problem. We use two semi-procedures, one searches for a $M' \in \set{M \in \nat^P \mid M(p)=0 \mbox{ for all } p \in P\setminus\set{p_f} \mbox{ and } M(p_f)\geq B} $ but not  in $\Reach{M_0}$ and the other one searches  a flat expression of transitions  $FE$ such that $\set{M \in \nat^P \mid M(p)=0 \mbox{ for all } p \in P\setminus\set{p_f}\mbox{ and } M(p_f)\geq B} \subseteq \Reach{M_0,L(FE)}$.

\begin{proposition}
  \label{prop:fixbound}
For a Petri net $\Net=\tuple{P,T,Pre,Post}$, a marking $M_0\in \nat^P$, a place $p_F \in P$ and a bound $B \in \nat$, testing whether $\set{M \in \nat^P \mid M(p)=0 \mbox{ for all } p \in P\setminus\set{p_f} \mbox{ and } M(p_f)\geq B} \subseteq \Reach{M_0}$ is decidable.
\end{proposition}

\begin{proof}
  The two semi-procedures to decide the inclusion are the following ones:
  \begin{enumerate}
  \item If $\set{M \in \nat^P \mid M(p)=0 \mbox{ for all } p \in P\setminus\set{p_f} \mbox{ and } M(p_f)\geq B} \not\subseteq \Reach{M_0}$ then there exists $b'\geq B$ such that $M' \notin \Reach{M_0}$ and $M'(p_f)=b'$ and $M'(p)=0 \mbox{ for all } p \in P\setminus\set{p_f}$. Hence a semi-procedure for non-inclusion enumerates such $b'$ and check  for non-reachability of  the marking $M'$.
   \item  If $\set{M \in \nat^P \mid M(p)=0 \mbox{ for all } p \in P\setminus\set{p_f} \mbox{ and } M(p_f)\geq B} \subseteq \Reach{M_0}$, then, from Proposition \ref{prop:semilinear:flat}, there exists a a flat expression $FE$ of transitions such that  $\set{M \in \nat^P \mid M(p)=0 \mbox{ for all } p \in P\setminus\set{p_f}\mbox{ and } M(p_f)\geq B} \subseteq \Reach{M_0,L(FE)}$ because the set  $\set{M \in \nat^P \mid M(p)=0 \mbox{ for all } p \in P\setminus\set{p_f}\mbox{ and } M(p_f)\geq B}$ is clearly semilinear (it can be described easily by a Presburger formula). Hence the semi-procedure for inclusion enumerates the flat expressions of transitions $FE$, computes the semilinear set $\Reach{M_0,L(FE)}$ thanks to Proposition \ref{prop:flatpetri} and tests whether $\set{M \in \nat^P \mid M(p)=0 \mbox{ for all } p \in P\setminus\set{p_f}\mbox{ and } M(p_f)\geq B} \subseteq \Reach{M_0,L(FE)}$ which amounts to test the inclusion of two semilinear sets which is decidable.
  \end{enumerate}
\end{proof}

\subsection{Finding the bound}

We now show why the single-place C.O.P. is decidable. Let $\Net=\tuple{P,T,Pre,Post}$  be a Petri net with a marking $M_0\in \nat^P$ and a place $p_f \in P$. One key aspect is that the set of markings reachable from $M_0$ with no token in the other places except $p_f$ is semilinear. This is a consequence of the following proposition.

\begin{proposition}\cite[Lemma IX.1]{leroux-presburger-lics13}
  \label{prop:reachperiodic}
Let $S \subseteq \nat^P$ be a semilinear set of markings. Then the set $\Reach{M_0}\cap S$ is a finite union of sets $\vect{b}+\vect{P}$ where $\vect{b} \in \nat^P$ and $\vect{P} \subseteq \nat^P$ is periodic.
\end{proposition}

\noindent From this proposition and Lemma \ref{lem:onedim:periodic}, we can deduce the following result.

\begin{proposition}
  \label{prop:intesect:semilinear}
 $\Reach{M_0} \cap \set{M \in \nat^P \mid M(p)=0 \mbox{ for all } p \in P\setminus\set{p_f}}$ is semilinear.
\end{proposition}

\begin{proof}
  From Proposition \ref{prop:reachperiodic}, we know that the set $\Reach{M_0} \cap \set{M \in \nat^P \mid M(p)=0 \mbox{ for all } p \in P\setminus\set{p_f}}$ is equal to $\bigcup_{1 \leq i \leq \ell} \vect{b}_i+\vect{P}_i$ where $\vect{b}_i \in \nat^P$ and $\vect{P}_i \subseteq \nat^P$ is periodic for each $i \in [1,\ell]$. Now note that by definition for each $p \in P \setminus \set{p_f}$, we have $\vect{b}_i(p)=0$ and for each element $\vect{v} \in \vect{P}_i$, we have $\vect{v}(p)=0$ for all $i \in [1,\ell]$. It means that the only relevant data in this union of sets is the projection over the place $p_f$. 
  From Lemma \ref{lem:onedim:periodic}, we hence have that each $\vect{P}_i$ is a semilinear set and as a direct consequence each $\vect{b}_i + \vect{P}_i$ is as well semilinear. \end{proof}

Another key point for the decidability of the single-place C.O.P. is the ability to test whether the intersection of the reachability set of a Petri net with a linear set is empty. In fact, it reduces to the reachability problem.

\begin{lemma}
  \label{lem:decid:intersect}
  If $S \subseteq \nat^P$ is a linear set of the form $\vect{b} + \vect{P}$ where $\vect{P}$ is finitely generated, then testing whether $\Reach{M_0} \cap S = \emptyset$ is decidable.
\end{lemma}

\begin{proof}
  We assume $\mathbf{P}=\set{\lambda_1.\vect{v}_1+\ldots+\lambda_k.\vect{v_k} \mid \lambda_i \in \nat \mbox{ for all } i \in [1,k]}$.
  From $\Net=\tuple{P,T,Pre,Post}$, we build another Petri net $\Net'=\tuple{P',T',Pre',Post'}$ such that:
  \begin{itemize}
  \item $P'=P \cup \set{p_{sim},p_{lin}}$,
  \item $T'=T \cup \set{t_{lin},t_{cons_1},\ldots,t_{cons_k},t_{end}}$.
  \end{itemize}

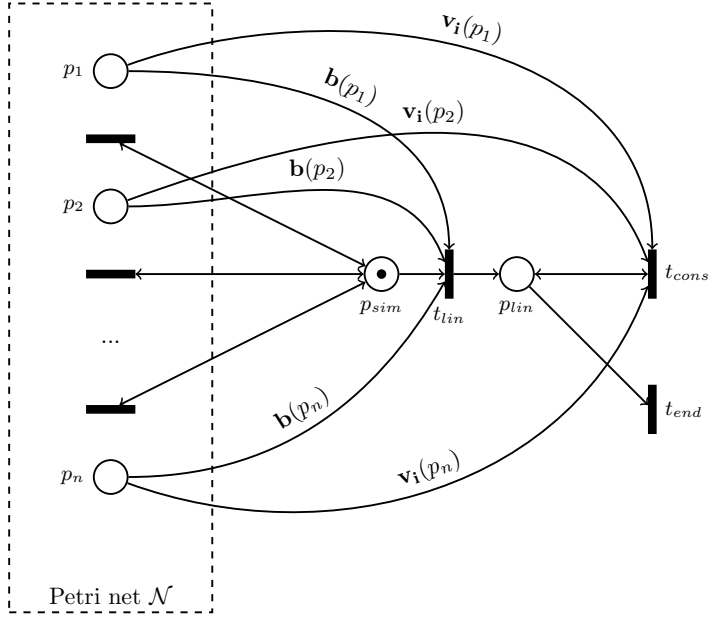
\begin{figure}[htbp]
\begin{center}
\scalebox{0.9}{
\begin{tikzpicture}[node distance=1cm]
\tikzstyle{every state}=[inner sep=3pt,minimum size=20pt]

\node(p1)[petri-p,draw=black,label=180:{\small $p_1$}]{};
\node(t1)[petri-t,draw=black,fill=black,below of=p1,label=180:{\small}]{};
\node(p2)[petri-p,below of =t1,draw=black,label=180:{\small $p_2$}]{};
\node(t2)[petri-t,draw=black,fill=black,below of=p2,label=180:{\small}]{};
\node(d)[below of=t2]{...};
\node(t3)[petri-t,draw=black,fill=black,below of=d,label=180:{\small}]{};
\node(p3)[petri-p,draw=black,below of=t3,label=180:{\small $p_n$}]{};
\node(r)[rectangle,thick,inner sep=0pt, minimum width=3cm,minimum height=9cm,draw=black,dashed,left of=p1,xshift=1cm,yshift=-3.5cm]{};
\node(l)[below of=p3,yshift=-0.75cm]{Petri net $\Net$};

\node(psim)[petri-p,draw=black,right of=t2,xshift=3cm,label=-90:{\small $p_{sim}$}]{}
edge[<->,thick] (t1)
edge[<->,thick] (t2)
edge[<->,thick] (t3);
\node(tlin)[petri-t2,draw=black,fill=black,right of=psim,xshift=0cm,label=-90:{\small
  $t_{lin}$}]{}
edge[<-,thick] (psim)
edge[<-,thick,out=90,in=0] node [sloped,above]{$\vect{b}(p_1)$}(p1)
edge[<-,thick,out=110,in=0] node [sloped,above]{$\vect{b}(p_2)$} (p2)
edge[<-,thick,out=-120,in=0] node [sloped,above]{$\vect{b}(p_n)$}  (p3);
\node(plin)[petri-p,draw=black,right of=tlin,label=-90:{\small $p_{lin}$}]{}
edge[<-,thick] (tlin);
\node(tcons)[petri-t2,draw=black,fill=black,right of=plin,xshift=1cm,label=0:{\small
  $t_{cons_i}$}]{}
edge[<->,thick] (plin)
edge[<-,thick,out=90,in=20] node [sloped,above]{$\vect{v_i}(p_1)$}(p1)
edge[<-,thick,out=110,in=20] node [sloped,above]{$\vect{v_i}(p_2)$} (p2)
edge[<-,thick,out=-110,in=-20] node [sloped,above]{$\vect{v_i}(p_n)$}  (p3);
\node(tend)[petri-t2,draw=black,fill=black,yshift=-1cm,below of=tcons,label=0:{\small
  $t_{end}$}]{}
edge[<-,thick] (plin);

\node[petri-tok] at (psim) {};

\end{tikzpicture}
}
\end{center}
\caption{The net $\Net'$ on a drawing (where only $t_{cons_i}$ is represented)}
\label{fig:intersect}
\end{figure}
  
  Intuitively, while there is a token in place $p_{sim}$ then $\Net'$ simulates $\Net$ (and let the token in $p_{sim}$). Then at some point $\Net'$ fires $t_{lin}$ which consumes the token in $p_{sim}$, consumes $\vect{b}[p]$ token in each place $p \in P$ and produces a token in $p_{lin}$. Then each transition $t_{cons_i}$, while there is a token in $p_{lin}$ (it tests the presence but does not consume it) consumes $\vect{p}_i(p)$ token in each place $p \in P$. Finally, the transition $t_{end}$ consumes the token in $t_{end}$ and does not produce any token.
  \begin{itemize}
  \item for all $t \in T$, we have $Pre'(t)(p_{sim})=1$, $Pre'(t)(p_{lin})=0$ and  $Pre'(t)(p)=Pre(t)(p)$ for all $p \in P$,
  \item for all $t \in T$, we have $Post'(t)(p_{sim})=1$, $Post'(t)(p_{lin})=0$ and  $Post'(t)(p)=Post(t)(p)$ for all $p \in P$,
   \item For what concerns the transition $t_{lin}$:
    \begin{itemize}
  \item $Pre'(t_{lin})(p_{sim})=1$, $Pre'(t_{lin})(p_{lin})=0$ and $Pre'(t_{lin})(p)=\vect{b}(p)$ for all $p \in P$,
  \item $Post'(t_{lin})(p_{sim})=0$, $Post'(t_{lin})(p_{lin})=1$ and $Post'(t_{lin})(p)=0$ for all $p \in P$,
    \end{itemize}  
  \item For what concerns the transitions $t_{cons_i}$ for $i \in [1,k]$:
    \begin{itemize}
  \item $Pre'(t_{cons_i})(p_{sim})=0$, $Pre'(t_{cons_i})(p_{lin})=1$ and $Pre'(t_{cons_i})(p)=\vect{p}_i(p)$ for all $p \in P$,
  \item $Post'(t_{cons_i})(p_{sim})=0$, $Post'(t_{cons_i})(p_{lin})=1$ and $Post'(t_{cons_i})(p)=0$ for all $p \in P$,
    \end{itemize}
  \item For what concerns the transition $t_{end}$ for $i \in [1,k]$:
    \begin{itemize}
     \item $Pre'(t_{end})(p_{sim})=0$, $Pre'(t_{end})(p_{lin})=1$ and $Pre'(t_{end})(p)=0$ for all $p \in P$,
     \item $Post'(t_{end})(p_{sim})=0$, $Post'(t_{end})(p_{lin})=0$ and $Post'(t_{cons_i})(p)=0$ for all $p \in P$.
       \end{itemize}
  \end{itemize}

  If we consider the marking $M'_0$ such that $M'_0(p_{sim})=1$, $M'_0(p_{lin})=0$ and $M'_0(p)=M_0(p)$ for all $p \in P$, then one can easily check that $\Reach{M_0} \cap S = \emptyset$ in $\Net$ iff $\vect{0} \notin \Reach{M'_0}$ in $\Net'$. In fact, the Petri net $\Net'$ first guesses non deterministically a marking of $\Reach{M_0}$ in the simulation phase, then it checks thanks to the transitions $t_{lin},t_{cons_1},\ldots,t_{cons_k}$ that this marking belongs to $S$ and finally it ends the test with $t_{end}$ which takes the token in $p_{lin}$.  The decidability of the reachability problem in Petri nets allows us to conclude.
\end{proof}

The previous results allow us to design two semi-procedures to decide the single place C.O.P. The first one enumerates the $B \in \nat$ and uses the result of Proposition \ref{prop:fixbound} to check if one is a cut-off. The other  one uses the fact that if there does not exist a cut-off then the set $\set{M \notin \Reach{M_0} \mid M(p)=0 \mbox{ for all } p \in P\setminus\set{p_f}}$ is semi-linear (by Proposition \ref{prop:intesect:semilinear})  and infinite and it includes a semi-linear set of the form $\set{\vect{b} +\lambda.\vect{p} \mid \lambda \in \nat}$ with $\vect{b},\vect{p} \in \nat^P$ and $\vect{0}<\vect{p}$. In this latter case we have $\Reach{M_0} \cap \set{\vect{b} +\lambda.\vect{p} \mid \lambda \in \nat} =\emptyset$ and we use  the result of Lemma \ref{lem:decid:intersect} to enumerate the $\vect{b},\vect{p}$ and find a pair satisfying this property.

\begin{theorem}
\label{thm:single:decid}
  The single place C.O.P. is decidable.  
\end{theorem}

\begin{proof}
  We consider a Petri net $\tuple{P,T,Pre,Post}$, an initial marking $M_0$ and a place $p_f$. We solve the single-place C.O.P. with the two following semi-procedures:
  \begin{enumerate}
    \item If there exists $B \in \nat$ such that for
      all $n \geq B$, we have $M^{(n)} \in \Reach{M_0}$ where $M^{(n)}$ is a marking verifying $M^{(n)}(p_f)=n$ and $M^{(n)}(p)=0$ for all $p \in P\setminus\set{p_f}$, the first semi-procedure enumerates the $b'$ of $\nat$ and tests whether for all $n \geq b'$, we have $M^{(n)} \in \Reach{M_0}$. According to Proposition \ref{prop:fixbound}, this test is possible and hence eventually the procedure finds $B$.
    \item Assume there does not exist $B \in \nat$ such that for
      all $n \geq B$, we have $M^{(n)} \in \Reach{M_0}$ where $M^{(n)}$ is a marking verifying $M^{(n)}(p_f)=n$ and $M^{(n)}(p)=0$ for all $p \in P\setminus\set{p_f}$. Let $F=\Reach{M_0} \cap \set{M \in \nat^P \mid M(p)=0 \mbox{ for all } p \in P\setminus\set{p_f}}$. By Proposition \ref{prop:intesect:semilinear}, this set is semilinear. Hence $\widetilde{F}= \set{M \in \nat^P \mid M(p)=0 \mbox{ for all } p \in P\setminus\set{p_f}}\setminus F$ is as well semilinear. Furthermore, the hypothesis holds iff $\widetilde{F}$ is infinite. As a consequence, there exists $\vect{b} \in \nat^P$ and a period $\vect{p} \in \nat^P$ such that $\vect{p} > \vect{0}$ and $\set{\vect{b} +\lambda.\vect{p} \mid \lambda \in \nat} \subseteq \widetilde{F}$. In that case, we have $\set{\vect{b} +\lambda.\vect{p} \mid \lambda \in \nat} \cap \Reach{M_0}=\emptyset$. Hence the second semi-procedure enumerates such two vectors $\vect{b}$ and $\vect{p}$ in $\nat^P$ until $\set{\vect{b} +\lambda.\vect{p} \mid \lambda \in \nat} \cap \Reach{M_0}=\emptyset$. This test can be performed thanks to Lemma \ref{lem:decid:intersect}.
  \end{enumerate}
\end{proof}

Thanks to Proposition \ref{prop:COPtosingleplace}, we obtain the result which concludes this section.

\begin{corollary}
  \label{cor:cop:decid}
  The C.O.P. is decidable.
\end{corollary}

\section{The specific case of symmetric rendez-vous}
\label{sec:symmetric}

Even though  the C.O.P. is decidable, the lower bound is quite bad as mentioned in Theorem \ref{thm:nonelem} and  the decision procedure presented in the proof of Theorem \ref{thm:single:decid} is quite technical. We show here that for a specific family of rendez-vous protocols, solving C.O.P. is easier.

\subsection{Definition and basic properties}

A  rendez-vous protocol $\Prot=\tuple{Q,Q_P,Q_L,\Sigma,q_i,q_f,\linebreak[0]q^L_i,q^L_f,E}$ is \emph{symmetric} if it respects the following property: for all $q,q' \in Q$ and $a \in \Sigma$, we have $(q,!a,q') \in E$ iff $(q,?a,q') \in E$. In this context we denote such transitions by $(q,a,q')$. We  furthermore assume w.l.o.g. that in the underlying graph of $\Prot$ for every states $q$ in $Q_P$ there is a path from $q_i$ to $q$ and a path from $q$ to $q_f$ (otherwise an initial configuration can never reach a configuration with a process in $q$ or from a configuration with a process in $q$ a final configuration can never been reached). We now work under these hypotheses.

In symmetric rendez-vous protocols, it is always possible to bring in any state as many pairs of processes one desires from the initial state $q_i$ and to remove as many pairs of processes (and bring them to the final state $q_f$). To perform such actions, it is enough to move pairs of processes following the same path (as the rendez-vous are symmetric, this is allowed by the semantics of rendez-vous protocols). We now state these properties formally. Let $\Prot=\tuple{Q,Q_P,Q_L,\Sigma,q_i,q_f,\linebreak[0]q^L_i,q^L_f,E}$ be a symmetric rendez-vous protocol.

\begin{lemma}
  \label{lem:pumpevenodd}
  Let $C \in \Confs$  verifying $C^{(|C|-1)}_i \trans^\ast C$. Then:
  \begin{enumerate}
  \item for all $C' \in \Confs$ such that $C(q) \leq C'(q)$ and $(C(q) = C'(q))\mod 2$ for all $q \in Q$, we have $C^{(|C'|-1)}_i \trans^\ast C'$,and,
  \item for all $C' \in \Confs$ such that $|C'|=|C|$ and $C'(q)\leq C(q)$ for all $q \in Q\setminus\set{q_f}$ and $(C(q) = C'(q)) \mod 2$ for all $q \in Q$, we have $C^{(|C'|-1)}_i \trans^\ast C'$.
  \end{enumerate}
\end{lemma}

\begin{proof}
  To prove Point 1, we consider $C' \in \Confs$ such that $C(q)
  \leq C'(q)$ and $(C(q) = C'(q)) \mod 2$ for all $q \in Q$. And we
  let $n=|C|-1$ and $m=|C'|-1$. First note that $n \leq m$. We want to
  show that $C^{(m)}_i \trans^\ast C'$. To do this we first execute
  from $C^{(m)}_i$ the same set of actions as in the execution $C^{(n)}_i \trans^\ast C$. We reach then a configuration $C''$ having the following properties: $C''(q)=C(q)$ for all $q \in Q \setminus \set{q_i}$ and $C''(q_i)=C(q_i)+m-n$. Then for each $q \in Q$ such that $C(q) < C'(q)$, we can bring pairwise $C'(q)-C(q)$ processes from $q_i$ to $q$ following the path  from $q_i$ to $q$. This is possible because the considered protocol is symmetric. Note that $C'(q)-C(q)$ is necessarily even since $(C(q) = C'(q)) \mod 2$. This leads us to the configuration $C''$. To prove Point 2 we proceed similarly by bringing pairwise processes from a state $q$ to $q_f$. 
\end{proof}

As a consequence, we show that there is a cut-off in $\Prot$ iff a final configuration with an even number and another one with an odd number of processes are reachable  in $\Prot$.

\begin{lemma}
  \label{lem:symevenodd}
   There exists $B \in \nat$ such that $C^{(n)}_i \trans^\ast C^{(n)}_f$ for
  all $n \geq B$ iff there exists an even $n_\even \in \nat$ and an odd $n_\odd \in \nat$ such that $C^{(n_\even)}_i \trans^\ast C^{(n_\even)}_f$ and  $C^{(n_\odd)}_i \trans^\ast C^{(n_\odd)}_f$.
\end{lemma}

\begin{proof}
  First obviously if there exists $B \in \nat$ such that for
  all $n \geq B$, we have $C^{(n)}_i \trans^\ast C^{(n)}_f$ then there exists an even natural $n_\even$ and an odd natural $n_\odd$ such that $C^{(n_\even)}_i \trans^\ast C^{(n_\even)}_f$ and  $C^{(n_\odd)}_i \trans^\ast C^{(n_\odd)}_f$. We are hence interested in showing the other direction. Assume there exists an even natural $n_\even$ and an odd natural $n_\odd$ such that $C^{(n_\even)}_i \trans^\ast C^{(n_\even)}_f$ and  $C^{(n_\odd)}_i \trans^\ast C^{(n_\odd)}_f$. Let $B=\max(n_\even,n_\odd)$ and $n \geq B$.  Suppose $n$ is even. Since $C^{(n_\even)}_i \trans^\ast C^{(n_\even)}_f$ and since $C^{(n)}_f$ is such that $C^{(n)}_f(q)= C^{(n_\even)}_f(q)$ for all $q \in Q \setminus \set{q_f}$ and $C^{(n_\even)}(q_f) \leq C^{n}(q_f)$, using 1. from Lemma \ref{lem:pumpevenodd}, we have $C^{(n)}_i \trans^\ast C^{(n)}_f$. The same technique applies when $n$ is odd.
\end{proof}

\subsection{The even-odd abstraction}
\label{subsec:evenodd}

We now present our tool to decide C.O.P. for a symmetric rendez-vous protocol $\Prot=\tuple{Q,Q_P,Q_L,\Sigma,q_i,q_f,\linebreak[0]q^L_i,q^L_f,E}$. We build an abstraction of the transition system $(\Confs,\trans)$ where we only remember the state of the leader and whether the number of processes in each state is even (denoted by  $\even$) or odd ($\odd$). Let $\widehat{\even}=\odd$ and $\widehat{\widehat{\even}}=\even$. The set of even-odd configurations is $\Confseo=Q_L \times \set{\even,\odd}^{Q_P}$. To an even-odd configuration $(q^L,\gamma) \in \Confseo$, we associate the set of configurations $\Interpb{(q^L,\gamma)} \subseteq \Confs$ such that $\Interpb{(q^L,\gamma)} = \set{C \in \Confs \mid C(q^L)=1 \mbox{ and } C(q)=0 \mod 2 \mbox{ iff } \gamma(q)=\even}$. We now define the even-odd  transition relation $\eotrans \subseteq \Confseo \times E \times E \times \Confseo$. We have $(q_1^L,\gamma_1) \leotrans{e,e'} (q_2^L,\gamma_2)$ iff one the following conditions holds:
\begin{enumerate}
\item $e=(q^L_1,a,q^L_2)$ and $e'=(q_1,a,q_2)$ belongs to $Q_P \times \Rv{\Sigma} \times Q_P$ and if $q_1= q_2$ then $\gamma_2 = \gamma_1$ else $\gamma_2(q_1)=\widehat{\gamma_1(q_1)}$, $\gamma_2(q_2)=\widehat{\gamma_1(q_2)}$ and $\gamma_2(q)=\gamma_1(q)$ for all $q \in Q_P \setminus \set{q_1,q_2}$.
\item $e,e' \in Q_P \times \Rv{\Sigma} \times Q_P$ and $q^L_1=q^L_2$ and $e=(q_1,a,q_2)$ and $e'=(q_3,a,q_4)$ and there exists $\gamma' \in \set{\even,\odd}^{Q_P}$ such that:
    \begin{itemize}
    \item if $q_1=q_2$ then $\gamma'=\gamma_1$ else $\gamma'(q_1)=\widehat{\gamma_1(q_1)}$, $\gamma'(q_2)=\widehat{\gamma_1(q_2)}$ and $\gamma'(q)=\gamma_1(q)$ for all $q \in Q_P \setminus \set{q_1,q_2}$, and,
    \item if $q_3=q_4$ then  $\gamma_2=\gamma'$ else  $\gamma_2(q_3)=\widehat{\gamma'(q_3)}$, $\gamma_2(q_4)=\widehat{\gamma'(q_4)}$ and $\gamma_2(q)=\gamma'(q)$ for all $q \in Q_P \setminus \set{q_3,q_4}$.
      \end{itemize}
    \end{enumerate}
    The relation $\leotrans{e,e'}$ reflects how the parity of the number of processes changes when performing a rendez-vous involving edges $e$ and $e'$. For instance, the first case illustrates a rendez-vous between the leader and a process, hence the parity of the number of states in $q_1$ and in $q_2$ changes except when these two control states are equal. The second case deals with a rendez-vous between two processes and it is cut in two steps to take care of  the cases like for instance $q_1\neq q_2$ and $q_3\neq q_4$ and $q_1 \neq q_4$ and $q_2=q_3$; in fact here the parity of the number of processes in $q_2$ should not change, since the first transition adds one process to $q_2$ and the second one removes one from it. We write $(q_1^L,\gamma_1) \eotrans (q_2^L,\gamma_2)$ iff there exists $e,e' \in E$ such that $(q_1^L,\gamma_1) \leotrans{e,e'} (q_2^L,\gamma_2)$ and $\eotrans^\ast$ denotes the reflexive and transitive closure of $\eotrans$.

    As said earlier, $(\Confseo,\eotrans)$ is an abstraction of $(\Confs,\trans)$. We will prove that this abstraction is enough to solve the C.O.P. For this, we define the following abstract configurations in $\Confseo$:
\begin{itemize}
\item $(q^L_i,\gamma^\even_i)$ and $(q^L_f,\gamma^\even_f)$ are such that $\gamma^\even_i(q)=\gamma^\even_f(q)=\even$ for all $q \in Q_P$;
\item $(q^L_i,\gamma^\odd_i)$ and $(q^L_f,\gamma^\odd_f)$ are such that $\gamma^\odd_i(q)=\gamma^\odd_f(q)=\even$ for all $q \in Q_P\setminus \set{q_i,q_f}$   and $\gamma^\odd_i(q_f)=\gamma^\odd_f(q_i)=\even$ and $\gamma^\odd_i(q_i)=\gamma^\odd_f(q_f)=\odd$.
\end{itemize}
Note that we have then $\set{C^{(n)}_i \mid n \mbox{ is even}} \subseteq \Interp{(q^L_i,\gamma^\even_i)}$ and $\set{C^{(n)}_i \mid n \mbox{ is odd}} \subseteq \Interp{(q^L_i,\gamma^\odd_i)}$ and $\set{C^{(n)}_f \mid n \mbox{ is even}} \subseteq \Interp{(q^L_f,\gamma^\even_f)}$ and $\set{C^{(n)}_f \mid n \mbox{ is odd}} \subseteq \Interp{(q^L_f,\gamma^\odd_f)}$. According to the definitions of the relations $\trans$ and $\eotrans$, we can easily deduce this first result.

    \begin{lemma}[Completeness]
      \label{lem:eocomplete}
Let $n \in \nat$. If $C^{(n)}_i \trans^\ast C^{(n)}_f$ and $n$ is even [resp. $n$ is odd]  then $(q^L_i,\gamma^\even_i) \eotrans^\ast(q^L_f,\gamma^\even_f) $ [resp. $(q^L_f,\gamma^\odd_i) \eotrans^\ast(q^L_f,\gamma^\odd_f)$].
\end{lemma}

The two next lemmas show that our abstraction is sound for C.O.P. The first one  can be proved by induction on the length of the path in $(\Confseo,\eotrans)$ using Point 1. of Lemma \ref{lem:pumpevenodd}. 

\begin{lemma}
  \label{lem:presound}
  If  $(q^L_i,\gamma^\even_i) \eotrans^\ast(q^L,\gamma) $ [resp.  $(q^L_i,\gamma^\odd_i) \eotrans^\ast(q^L,\gamma) $] then there exists $n \in \nat\setminus\set{0}$ such that $n$ is even [resp. $n$ is odd] and $C^{(n)}_i \trans^\ast C$ with $C \in \Interp{(q^L,\gamma)}$.
\end{lemma}

\begin{proof}
  Assume $(q^L_i,\gamma^\even_i) \leotrans{e_1,e'_1} (q^L_1,\gamma_1) \leotrans{e_2,e'_2}(q^L_2,\gamma_2)\leotrans{e_3,e'_3} \ldots \leotrans{e_k,e'_k} (q^L_k,\gamma_k)$ with $(q^L_k,\gamma_k)=(q^L,\gamma)$ . We reason by induction on $k$. For $k=0$, we have $C^{(2)}_i \in \Interp{(q^L_i,\gamma^\even_i)}$, hence the property holds. Now suppose $k>1$ and that the property holds for $k-1$. Hence there exists $C' \in \Interp{(q^L_{k-1},\gamma^\even_{k-1})}$ and $n \in \nat\setminus\set{0}$ such that $n$ is even and $C^{(n)}_i \trans^\ast C'$. We have two cases:
  \begin{enumerate}
  \item $e_k=(q^L_{k-1},a,q^L_{k})$ and $e'_k=(q_1,a,q_2)$ (in other words the pair $(e_k,e'_k)$ involves a transition of the leader). Then to take this rendez-vous from $C'$, we need to have $C'(q_1)>0$ but it might not be the case. However by 1. of Lemma \ref{lem:pumpevenodd}, if we consider the configuration $C''$ such that $C''(q)=C'(q)$ for all $q \in Q \setminus \set{q_1}$ and $C''(q_1)=C'(q_1)+2$ then $C^{(n+2)}_i \trans^\ast C''$. Note that by definition $C'' \in \Interp{(q^L_{k-1},\gamma^\even_{k-1})}$. From $C''$ the rendez-vous between edges $e_k$ and $e'_k$ can take place and it leads to a configuration $C$, hence $C^{(n+2)}_i \trans^\ast C'' \trans C$,  and by definition of $\eotrans$ we have necessarily that $C \in \Interp{(q^L,\gamma)}$.
  \item The case where $e_k,e'_k \in Q_P \times \Rv{\Sigma} \times Q_P$ can be treated similarly always thanks to Point 1. of Lemma \ref{lem:pumpevenodd}.
  \end{enumerate}
The proof for the  case where $n$ is odd is identical.
\end{proof}

Using Point 2. of Lemma \ref{lem:pumpevenodd} we obtain the soundness of our abstraction.

\begin{lemma}[Soundness]
  \label{lem:eosound}
   If  $(q^L_i,\gamma^\even_i) \eotrans^\ast(q^L_f,\gamma^\even_f) $ [resp.  $(q^L_i,\gamma^\odd_i) \eotrans^\ast(q^L_f,\gamma^\odd_f) $] then there exists $n \in \nat$ such that $n$ is even [resp. $n$ is odd] and $C^{(n)}_i \trans^\ast C^{(n)}_f$.
 \end{lemma}

 \begin{proof}
Assume $(q^L_i,\gamma^\even_i) \eotrans^\ast(q^L_f,\gamma^\even_f)$. Then thanks to Lemma \ref{lem:presound}, we know that there exists $n \in \nat \setminus \set{0}$ such that $n$ is even and $C^{(n)}_i \trans^\ast C$ with $C \in \Interp{(q^L_f,\gamma^\even_f)}$. Note that by definition of $\trans$, we have $|C|=n+1$. Consider the configuration $C'$ such that $C'(q^L_f)=1$ and $C'(q_f)=n$ and $C'(q)=0$ for all $q \in Q \setminus \set{q^L_f,q_f}$, then using Point 2. of Lemma \ref{lem:pumpevenodd}, we have $C^{(n)}_i \trans^\ast C'$ and $C'=  C^{(n)}_f$. The case where $n$ is odd can be treated similarly.
\end{proof}

Thanks to the Lemmas \ref{lem:symevenodd}, \ref{lem:eocomplete} and \ref{lem:eosound} to solve the C.O.P. when the considered rendez-vous protocol is symmetric it is enough to check whether $(q^L_i,\gamma^\even_i) \eotrans^\ast(q^L_f,\gamma^\even_f) $ and $(q^L_i,\gamma^\odd_i) \eotrans^\ast(q^L_f,\gamma^\odd_f) $. But since the transition system $(\Confseo,\eotrans)$ has a finite number of vertices whose number is bounded by $|Q_L|\cdot 2^{|Q_P|}$, these two reachability questions can be solved in \textsc{NPspace} in $|Q|$. By Savitch's  theorem, we obtain the following result.

\begin{theorem}
  \label{thm:pspace:symmetric}
C.O.P. restricted to symmetric rendez-vous protocols is in \textsc{PSpace}.
  \end{theorem}

\section{Supressing the leader}
\label{sec:noleader}

\subsection{Definition and properties}

A rendez-vous protocol
$\Prot=\tuple{Q,Q_P,Q_L,\Sigma,q_i,q_f,\linebreak[0]q^L_i,q^L_f,E}$
has \emph{no leader} when $Q_L=\set{q^L_f}$ and $q^L_i=q^L_f$ and the
transition relation does not refer to the state in $Q_L$, i.e. $E
\subseteq Q_P \times \Rv{\Sigma} \times Q_P$.  We can then
assume that $\Prot=\tuple{Q_P,\Sigma,q_i,q_f,E}$ and 
delete any reference to the leader state. We suppose 
again w.l.o.g. that in the considered rendez-vous
protocols without leader there is a path
from $q_i$ to $q$ and a path from $q$ to $q_f$  for all $q$ in $Q_P$. Rendez-vous protocols
with no leader
enjoy some properties easing the resolution of the  C.O.P.

\begin{lemma}
\label{lem:ntonplsuone}
  Let $\Prot=\tuple{Q_P,\Sigma,q_i,q_f,E}$ be  a rendez-vous protocol
  with  no leader. Then the following properties hold:
  \begin{enumerate}
   \item If $C^{(n)}_i \trans^\ast C^{(n)}_f$ and $C^{(m)}_i
     \trans^\ast C^{(m)}_f$ for $m,n \in \nat$, then  $C^{(n+m)}_i \trans^\ast C^{(n+m)}_f$.
 \item  There exists $B \in \nat$ such that $C^{(n)}_i \trans^\ast C^{(n)}_f$ for
  all $n \geq B$ iff there
  exists $N \in \nat$ such that  $C^{(N)}_i \trans^\ast C^{(N)}_f$ and
  $C^{(N+1)}_i \trans^\ast C^{(N+1)}_f$.
  \end{enumerate}
\end{lemma}

\begin{proof}
  \begin{enumerate}
  \item This point is  a direct consequence of the semantics of
  rendez-vous protocols associated with the fact that there is no
  leader. In fact assume $C^{(n)}_i \trans^\ast C^{(n)}_f$ and $C^{(m)}_i
     \trans^\ast C^{(m)}_f$. And consider the configuration $C$ such
     that $C(q_i)=m$, $C(q_f)=n$  and $C(q)=0$ for all $q \in Q_P
     \setminus \set{q_i,q_f}$. Then it is clear that we have
     $C^{(n+m)}_i \trans^\ast C \trans^{\ast} C^{(n+m)}_f$,  the first
     part of this execution mimicking the execution $C^{(n)}_i
     \trans^\ast C^{(n)}_f$ and the last part mimics the execution $C^{(m)}_i
     \trans^\ast C^{(m)}_f$ on the $m$ processes left in $q_i$ in $C$.
  \item If  there exists $B \in \nat$ such that $C^{(n)}_i \trans^\ast
    C^{(n)}_f$  for
  all $n \geq B$, then we
  have $C^{(B)}_i \trans^\ast C^{(B)}_f$ and $C^{(B+1)}_i \trans^\ast
  C^{(B+1)}_f$. Assume now that there exists $N \in \nat$ such that
  $C^{(N)}_i \trans^\ast C^{(N)}_f$ and $C^{(N+1)}_i \trans^\ast
  C^{(N+1)}_f$. We show that for all $n \geq N^2$, we have
  $C^{(n)}_i \trans^\ast C^{(n)}_f$. Let $n \geq N^2$ and let $R \in
  [0,N-1]$ be such that $(n=R) \mod N$. By definition of the modulo,
  there exists $A \geq 0$ such that $n=A\cdot N + R$.  Since $ n \geq
  N^2$, we have necessarily $A \geq N$.  As a consequence we can rewrite $n$ as: $n = R\cdot(N+1) +(A-R)\cdot  N$. But then since
  $C^{(N)}_i \trans^\ast C^{(N)}_f$, by 1. we have $C^{((A-R)\cdot N)}_i
  \trans^\ast C^{((A-R)\cdot N)}_f$ and since  $C^{(N+1)}_i \trans^\ast
  C^{(N+1)}_f$, by 1. we have $C^{(R \cdot(N+1))}_i \trans^\ast
  C^{(R \cdot (N+1))}_f$. By a last application of 1. we get $C^{(n)}_i \trans^\ast C^{(n)}_f$.
  \end{enumerate}
 
\end{proof}

\subsection{The symmetric case}

We will now see how the procedure proposed in the proof of Theorem
\ref{thm:pspace:symmetric} to solve in polynomial space the C.O.P. for
symmetric rendez-vous protocols can be simplified when there is no leader.   Let
$\Prot=\tuple{Q_P,\Sigma,q_i,q_f,E}$ be  a symmetric rendez-vous protocol with
no leader and let $(\Confseo,\eotrans)$ be the abstract transition
system of $(\Confs,\trans)$ as defined in Section
\ref{subsec:evenodd}. If we adapt the results of  Lemmas
\ref{lem:symevenodd}, \ref{lem:eocomplete} and \ref{lem:eosound} to
the no leader case, we deduce that to solve the C.O.P. it is enough to
check whether $\gamma^\even_i \eotrans^\ast \gamma^\even_f $ and
$\gamma^\odd_i \eotrans^\ast \gamma^\odd_f $ (we have deleted the
leader states from these results). Note that  by definition
$\gamma^\even_i=\gamma^\even_f$, hence the only
thing to verify is if  $\gamma^\odd_i
\eotrans^\ast\gamma^\odd_f $ holds. This check can be made efficiently
using the fact that there is no leader, because any  reodering of a
path is still a path in $(\Confseo,\eotrans)$  (since 
we do not need to worry anymore about the leader state) and we can
delete the pairs of edges that consecutively repeat since they have  the same action
on the parity.

\begin{lemma}
  \label{lem:shortpath}
  If $\gamma  \eotrans^\ast \gamma'$ then  there exists $k  \leq
  |E|^2$ and $e_1,e'_1,e_2,e'_2,\ldots,e_k,\linebreak[0]e'_k \in E$
  such that  $\gamma \leotrans{e_1,e'_1} \gamma_1 \leotrans{e_2,e'_2} \ldots
\leotrans{e_k,e'_k} \gamma'$.
\end{lemma}

\begin{proof}
Assume $\gamma \leotrans{e_1,e'_1} \gamma_1 \leotrans{e_2,e'_2} \ldots
\leotrans{e_k,e'_k} \gamma'$  with $k > |E|^2$. Consequently there
exists $i,j \in [1,k]$ such that $i \neq j$ and $(e_i,e'_i)=(e_j,e'_j)$. Note that
according to the semantics of $\eotrans$, when there is no leader, if we have $\gamma''
\leotrans{e,e'} \gamma''_1 \leotrans{d,d'} \gamma''_2$
then  we also have  $\gamma''
\leotrans{d,d'} \gamma''_3 \leotrans{e,e'} \gamma''_2$  and furthermore
if $(e,e')=(d,d')$ then $\gamma'' =\gamma''_2$. As a consequence, we can
assume that $j=i+1$ (otherwise we can reorder the run) and that
$\gamma_i=\gamma_{i+2}$. This allows us to shorten the execution from
$\gamma$ to $\gamma'$ by deleting the edges
$\leotrans{e_i,e'_i}\leotrans{e_{i+1},e'_{i+1}}$. We can repeat this
  operation until we obtain a run of length strictly smaller that $|E|^2$.
\end{proof}

It means that if $\gamma^\odd_i
\eotrans^\ast\gamma^\odd_f $ then there is a path of polynomial
length (in the size of $\Prot$) between these two abstract
configurations. It is hence enough to guess such a sequence of
polynomial length and to check that it
effectively corresponds to a path in $(\Confseo,\eotrans)$.

\begin{theorem}
  C.O.P. for symmetric rendez-vous protocols with no leader is in \textsc{NP}.
  \end{theorem}

\subsection{Upper bound for the C.O.P. with no leader}

We now prove that the C.O.P. for rendez-vous protocols with no
leader reduces to the reversible reachability problem in Petri
nets. Let $\Prot=\tuple{Q_P,\Sigma,q_i,q_f,E}$  be a rendez-vous protocol
with no leader and such that w.l.o.g. there is no edge going out of
$q_f$\footnote{To achieve this, we can simply duplicate $q_f$ adding a new final state
$q'_f$ and for each edge going into $q_f$ we add an edge from
the same state to $q'_f$}.

Let  $\Net_\Prot=\tuple{P,T,Pre,Post}$ be the Petri net whose construction is
provided in Section \ref{sec:topetri} (where we have removed all the
places corresponding to leader states as well as the transition
$t^L_f$). From $\Net_\Prot $, we build the reverse Petri net $\Net^R_\Prot$
obtained  by keeping the same set of places and reversing all the
transitions. Formally  $\Net^R_\Prot=\tuple{P^R,T^R,Pre^R,Post^R}$,
where $P^R=\set{p^R \mid p \in P}$, $T^R=\set{t^R \mid t \in T}$ and
for all $p^R \in P^R$ and $t^R \in T^R$, we have
$Pre^R(t^R)(p^R)=Post(t)(p)$ and $Post^R(t^R)(p^R)=Pre(t)(p)$. Let
$M^R_0$ be the marking such that $M^R_0(p^R)=0$ for all $p^R \in P^R$
and  $(M^{R,(n)}_f)_{\set{n \in \nat}}$ be the family of markings verifying
$M_f^{R,(n)}(p^R_{q_f})=n$ and $M_f^{R,(n)}(p)=0$ for all $p \in P^R \setminus
\set{p^R_{q_f}}$. A direct consequence of Lemma \ref{lem:reductionCOP}
and of the definition of $\Net^R_\Prot$ is that $C^{(n)}_i \trans^\ast C^{(n)}_f$ iff
$M^R_0 \in \Reach{M_f^{R,(n)}}$ for all $n \in \nat$.

From $\Net_\Prot$ and $\Net^R_\Prot$, we build the Petri net
$\Net'_\Prot$ obtained by taking the disjoint unions of places and
transitions of the two nets except for the place $p_{q_f}$ and
$p^R_{q_f}$ which are merged in a single place $p_{q_f}$. Formally,
$\Net'_\Prot=\tuple{P',T',Pre',Post'}$ where $P'=(P \cup P^R) \setminus
\set{p^R_{q_f}}$, $T'=T \cup T^R$, $Pre'(t)(p)= Pre(t)(p)$ and
 $Post'(t)(p)= Post(t)(p)$ and  $Pre'(t)(p^R)=Post'(t)(p^R)=0$ for all
$p \in P$, $p^R\in P^R$ and $t \in T$, $Pre'(t^R)(p^R)= Pre^R(t^R)(p^R)$ and
$Post'(t^R)(p^R)= Post^R(t^R)(p^R)$ and $Pre'(t^R)(p)=Post'(t^R)(p)=0$
for all $p^R \in P^R$, $p \in P \setminus\set{p_{q_f}}$ and $t \in T$,
and $Pre'(t^R)(p_{q_f})=Pre^R(t^R)(p^R_{q_f}))$ and
$Post'(t^R)(p_{q_f})=Post^R(t^R)(p^R_{q_f}))$  (this last case
corresponds to the merging of  $p_{q_f}$ and $p^R_{q_f}$). Figure
\ref{fig:exnolead} provides an example of this
latter Petri net.

\begin{figure}[tbp]
\begin{center}
\scalebox{1}{
\begin{tikzpicture}[node distance=1.5cm]
\tikzstyle{every state}=[inner sep=3pt,minimum size=20pt]
\node(0)[state]{$q_i$}
edge [<-,thick,loop above] node[auto] {$?a$}(0);
\node(1)[finalstate,right of=0,node distance=2cm]{$q_f$}
edge [<-,thick,out=-140,in=-40] node[auto] {$!a$}(0)
edge [<-,thick,out=140,in=40] node[auto,swap] {$!b$}(0)
edge [<-,thick] node[auto,swap] {$?b$}(0);

\node(t1)[petri-t,draw=black,fill=black,right of=1,xshift=0.5cm,label=180:{\small
  $a$}]{};
\node(t2)[petri-t,draw=black,fill=black,right of=t1,label=180:{\small
  $b$}]{};
\node(t3)[petri-t,draw=black,fill=black,right of=t2,label=180:{\small
  $b^R$}]{};
\node(t4)[petri-t,draw=black,fill=black,right of=t3,label=180:{\small
  $a^R$}]{};
\node(p1)[petri-p,above of=t1,xshift=0.75cm,draw=black,label=0:{\small
  $p_{q_i}$}]{}
edge[<-,thick,out=-140,in=120] (t1)
edge[->,thick]node[auto] {$2$} (t1)
edge[->,thick]node[auto] {$2$} (t2);
\node(t5)[petri-t2,draw=black,fill=black,left of=p1,label=180:{\small
  $t_i$}]{}
edge[->,thick](p1);
\node(p2)[petri-p,below of=t2,xshift=0.75cm,draw=black,label=180:{\small
  $p_{q_f}$}]{}
edge[<-,thick] (t1)
edge[<-,thick]node[auto,swap] {$2$} (t2)
edge[->,thick]node[auto] {$2$} (t3)
edge[->,thick](t4);
\node(p3)[petri-p,above of=t3,xshift=0.75cm,draw=black,label=180:{\small
  $p^R_{q_i}$}]{}
edge[<-,thick]node[auto] {$2$} (t3)
edge[<-,thick]node[auto,swap] {$2$} (t4)
edge[->,thick,out=-40,in=60]node[auto] {} (t4);

\node(t6)[petri-t2,draw=black,fill=black,right of=p3,label=0:{\small
  $t^R_i$}]{}
edge[<-,thick] (p3);
\end{tikzpicture}
}
\end{center}
\caption{A rendez-vous protocol with no leader $\Prot$ and the
  associated Petri net $\Net'_\Prot$}
\label{fig:exnolead}
\end{figure}
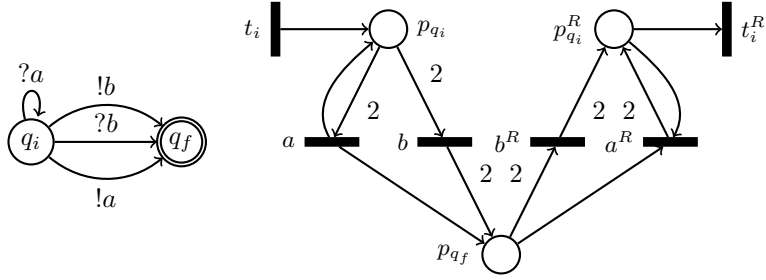

We now explain why this new net is useful to solve the  C.O.P. when
there is no leader. First remember that  thanks to Point 2. of Lemma
\ref{lem:ntonplsuone} it is enough to check whether there
  exists $N \in \nat$ such that  $C^{(N)}_i \trans^\ast C^{(N)}_f$ and
  $C^{(N+1)}_i \trans^\ast C^{(N+1)}_f$. Intuitively, in $\Net'_\Prot$
  this property will be witnessed  by the fact that we can bring $N+1$
  tokens in $p_{q_f}$ using transitions in $T$ and remove $N$ tokens
  from $p_{q_f}$ thanks to the transitions in $T^R$ letting hence one
  token in $p_{q_f}$ and similarly if there is already a token in
  $p_{q_f}$ we can bring $N$ others and remove afterwards $N+1$. As for
$\Net_\Prot$, we let $M_0$ be the marking with no token, and
$(M^{(n)})_{\set{n \in \nat}}$ be the family of markings such that
$M^{(n)}(p_{q_f})=n$ and $M^{(n)}(p)=0$ for all $p \in P' \setminus
\set{p_{q_f}}$. Note that since there is no leader, we have here $M_0=M^{(0)}$. The next lemma states the correctness of our reduction to the reversible
reachability problem.

\begin{lemma}
  \label{lem:correct-reverse}
There
  exists $N \in \nat$ such that  $C^{(N)}_i \trans^\ast C^{(N)}_f$ and
  $C^{(N+1)}_i \trans^\ast C^{(N+1)}_f$ iff $M^{(1)} \in \Reach{M_0}$
  and $M_0 \in \Reach{M^{(1)}}$ in the Petri net $\Net'_\Prot$.
\end{lemma}

\begin{proof}
Assume that
there exists $N \in \nat$ such that  $C^{(N)}_i \trans^\ast C^{(N)}_f$ and
  $C^{(N+1)}_i \trans^\ast C^{(N+1)}_f$ then in $\Net'_\Prot$ from
  $M^{(0)}$ we can reach $M^{(N+1)}$ taking only transitions in $T$ (thanks
  to Lemma \ref{lem:reductionCOP}) and from $M^{(N+1)}$ we can  reach
  $M^{(1)}$ letting one token in $p_{q_f}$ and removing all the other
  tokens using only transitions in $T^R$ (and again using Lemma
  \ref{lem:reductionCOP} and the fact that $C^{(N)}_i \trans^\ast
  C^{(N)}_f$)). Hence $M^{(1)} \in \Reach{M_0}$. Similarly we can show
  that $M_0 \in \Reach{M^{(1)}}$ by from $M^{(1)}$ reaching
  $M^{(N+1)}$ using transitions in $T$ and the fact that $C^{(N)}_i
  \trans^\ast C^{(N)}_f$ . And  then from $M^{(N+1)}$ we can reach
  $M_0$  using transitions in $T^R$ and the fact that $C^{(N+1)}_i
  \trans^\ast C^{(N+1)}_f$.

Assume now that $M^{(1)} \in \Reach{M_0}$. Note that in the execution
from $M_0$ to $M^{(1)}$, we can assume that first the only transitions
that occur are in $T$ and then the only used transitions belong to
$T^R$, because the only common place between these two sets of
transitions is $p_{q_f}$ and transitions from $T$ only produce tokens
in this place whereas transitions in $p_{q_f}$ only consume them
(remember we assume that in  $\Prot$ no transition goes out of
$q_f$). Hence in $\Net'$ we have an execution of  the form $M_0
\ltransP{t_0} \ldots  \ltransP{t_k} M \ltransP{t'_0}  \cdots
\ltransP{t'_\ell} M^{(1)}$ where $\set{t_0,\ldots,t_k} \subseteq T$ and
$\set{t'_0,\ldots,t'_\ell} \subseteq T^R$. Since the transitions in
$T$ only consume and produce tokens in $P$ and the one in $T^R$ only
consume tokens in $P^R  \cup \set{p_{q_f}}$ and produces tokens in
  $P^R$, we deduce that there exists some $N$ such that
  $M=M^{(N+1)}$. Using Lemma \ref{lem:reductionCOP}, we deduce from
  $M^{(N+1)}\in \Reach{M_0}$ that $C^{(N+1)}_i \trans^\ast
C^{(N+1)}_f$ and from the fact that $M^{(1)} \in \Reach{M^{(N+1)}}$ in
  $\Net'_\Prot$ using only transitions in $T^R$ that we
  have as well $M_0 \in  \Reach{M^{(N)}}$ in $\Net^R$ and consequently $C^{(N)}_i \trans^\ast
C^{(N)}_f$.
  \end{proof}

Since we know that the reversible reachability problem for Petri net
is \textsc{EXPspace}-complete \cite{leroux-vector-lmcs-13}, we obtain the following complexity
result.

\begin{theorem}
C.O.P. restricted to rendez-vous protocols with no leader is in \textsc{EXPSpace}.
\end{theorem}

\begin{figure}[hbtp]
\begin{center}
\scalebox{0.9}{
\begin{tikzpicture}[node distance=1.8cm]
\tikzstyle{every state}=[inner sep=3pt,minimum size=20pt]
\node(0)[state]{$q_i$};

\node(1)[state,right of=0]{$q_1$}
edge [<-,thick,out=-160,in=-20] node[sloped,above] {$?1$}(0)
edge [<-,thick,out=160,in=20] node[sloped,above] {$!1$}(0);

\node(2)[state,right of=1]{$q_2$}
edge [<-,thick] node[auto] {$!2$}(1)
edge [<-,thick,bend left=40] node[auto] {$?2$}(0);

\node(3)[state,right of=2]{$q_3$}
edge [<-,thick] node[auto] {$!3$}(2)
edge [<-,thick,bend left=50] node[auto] {$?3$}(0);

\node(4)[right of=3,xshift=-1cm]{$\cdots$};

\node(n)[state,right of=4]{$q_n$}
edge [<-,thick] node[auto] {$!n$}(4)
edge [<-,thick,bend left=60] node[auto] {$?n$}(0)
edge [<-,thick,loop above] node[auto] {$!a$}(n);

\node(f)[finalstate,right of=n]{$q_f$}
edge [<-,thick] node[auto] {$!a$}(n)
edge [<-,thick,bend left=-40] node[auto,swap] {$?a$}(0);
\end{tikzpicture}
}
\end{center}
\caption{A rendez-vous protocol with no leader and an exponential cut-off}
\label{fig:expcut}
\end{figure}
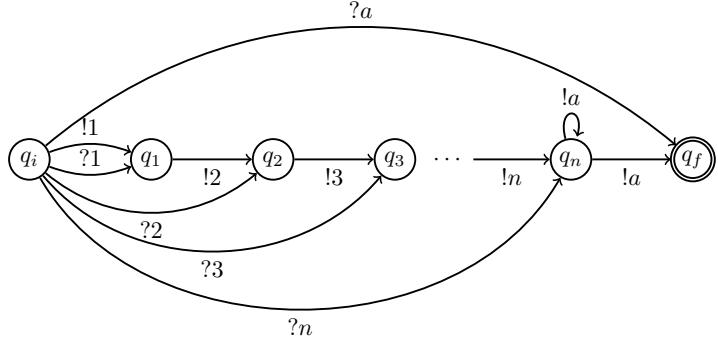

We were not able to propose a lower bound for the C.O.P. apart for the
general case, but when there is no leader, we know that there is a
protocol which admits a cut-off whose value is exponential in the size
of a protocol. This protocol is shown on Figure \ref{fig:expcut}. To
bring a process in $q_1$, we need in fact two processes, to bring a
process in $q_2$ and empty $q_1$, we need four processes and so
on. The letter $a$ is then used to ensure that as soon as we have
processes only in $q_n$ and in $q_i$ (and at least one of them in each
of these states), there is a way to bring all of
them in $q_f$.

\section{Conclusion}

We have shown here that the C.O.P. is decidable for rendez-vous networks. Furthermore we have provided complexity upper bounds when considering restrictions on the networks such as symmetric rendez-vous or absence of leader. Unfortunately, we did not succeed in finding matching lower bounds. Reducing other problems to the C.O.P. is in fact tedious without leader or when allowing only symmetric rendez-vous, because it is then quite hard to enforce that a specific number of processes are in some states which is a property that is in general needed to design reductions.  However we have some hope to either improve our upper bounds or find matching lower bounds. We wish as well to understand in which matters the techniques we used could be adapted to other parameterized systems and more specifically to population protocols. Finally, one of the justification to consider the cutoff problem is that in some distributed systems it could be the case that a correctness property does not hold for any number of processes, but that a minimal number of participants is needed to reach a goal. It could be interesting to study a variant of our cutoff problem where we do not require all the processes to reach a final state but we want to know given a number of processes how many among them can be brought  in such a state. An interesting property could be to check whether there exists a bound $b$ such that for any number of processes, the minimal number that can not be brought to a final state by any execution is always lower than $b$. In such networks, it would mean that at most $b$ entities have to be sacrificed to let the others reach the final state.

\bibliographystyle{plainurl}

\bibliography{cutoff.bib}


\end{document}